%% file: main.tex
\documentclass[10pt]{article}
\usepackage[lmargin=1.1in,rmargin=1.1in,tmargin=1in,bmargin=1in]{geometry}

\usepackage{listings}
\usepackage{url}
\usepackage{graphicx}
\usepackage{balance}  
\usepackage{amsmath,amssymb,amsthm,graphicx,multirow}
\usepackage[ruled,lined,noend,linesnumbered]{algorithm2e}
\SetArgSty{textrm}
\usepackage{enumitem}
\usepackage{verbatim}
\usepackage{bm}
\usepackage{pseudocode}
\usepackage{tabularx,times}
\usepackage{balance}
\usepackage{microtype}
\usepackage{color}
\usepackage{bbm}
\usepackage{framed}
\usepackage{balance}  
\newcolumntype{Y}{>{\centering\arraybackslash}X}

\let \originalleft \left
\let\originalright\right
\renewcommand{\left}{\mathopen{}\mathclose\bgroup\originalleft}
\renewcommand{\right}{\aftergroup\egroup\originalright}

\newcommand{\ctext}[1]{\texttt{\small #1}}

\newcommand{\leftopen}{left open set}
\newcommand{\rightopen}{right open set}
\newcommand{\opensets}{open sets}

\newcommand{\Tl}{T_L}
\newcommand{\Tr}{T_R}

\newcommand{\node}{\mbox{\tt Node}}
\newcommand{\leaf}{\mbox{\tt Leaf}}
\newcommand{\func}[1]{{\sc #1}}

\newcommand{\augvalueone}[1]{{\mathcal{A}(#1)}}

\newcommand{\rangesweep} {{prefix-combine trees}}
\newcommand{\segtree} {{segment tree}}
\newcommand{\symdiff} {{symetric difference sets}}

\newcommand{\countmap}{{{count map}}}
\newcommand{\mb}[1]{\mbox{\func{\textbf{#1}}}}
\newcommand{\bool}{\mb{Bool}}

\newcommand{\hide}[1]{} 
\newcommand{\mapstructure}[1]{{{\begin{center}\vspace{.0em}#1\end{center}}}} 

\newcommand{\para}[1]{\vspace{0.01in}\noindent\textbf{#1 }}

\newcommand{\lc}[1]{{{l(#1)}}}
\newcommand{\rc}[1]{{{r(#1)}}}
\newcommand{\funca}{{\rho}}
\newcommand{\funcb}{{\phi}}
\newcommand{\sweepstructure}{{prefix structure}}

\newtheorem{theorem}{Theorem}

\newtheorem{corollary}{Corollary}

\newcommand{\fullin}[1]{{\if 1<2
{#1}
\else
{}
\fi}}

\begin{document}
\title{Parallel Range, Segment and Rectangle Queries \\with Augmented Maps}

\author{
Yihan Sun\\
Carnegie Mellon University\\
yihans@cs.cmu.edu
\and
Guy E.\ Blelloch\\
Carnegie Mellon University\\
guyb@cs.cmu.edu}
\date{}



\maketitle

\input{abstract}

\input{intro}

\input{related}

\input{prelim}

\input{aug_map}

\input{sweepline}

\input{application}

\input{app-seg}
\input{app-rec}

\input{exp}

\input{conclusion}
\balance
\newpage
\bibliographystyle{plain}
\bibliography{main}
\appendix
\input{appendix}

\end{document}

%% file: abstract.tex
\begin{abstract}
\hide{
The range query problem and segment query problem are fundamental problems
in computational geometry, and also have extensive applications in many areas.
Despite the large body of theory works on parallel geometry algorithms,
efficient implementations of them can be complicated.
We know of very few
practical implementations of the theoretical works due to many delicate algorithmic details.
On the other hand most of the implementation-based
works do not have tight theoretical bounds.
In this paper we focus on the fast and concise (fewer lines of code) implementations of parallel algorithms for these problems with provable theoretical cost bound.
We propose to use a simple framework (the augmented map) to model these problem, and discuss possible implementations.
Based on augmented map interface and implementations, we can develop both multi-level tree structures and sweepline algorithms addressing range queries and segment queries. All of our data structures are work-efficient to build in theory (all need $O(n\log n)$ work) and achieve a low parallel depth (poly-logarithmic for the multi-level tree structures, and $O(n^{\epsilon})$ for sweepline algorithms). In addition, the implementation of such structures, on top of the augmented maps, is very simple and concise (about 100 lines of code), as well as fast in practice.
}
The range, segment and rectangle query problems are fundamental problems in
computational geometry, and have extensive applications in many
domains.  Despite the significant theoretical work on these problems,
efficient implementations can be complicated.  We
know of very few practical implementations of the algorithms in parallel, and most
implementations do not have tight theoretical bounds.  In this paper,
we focus on simple and efficient parallel algorithms and
implementations for range, segment and rectangle queries, which have tight worst-case bound in theory and good parallel performance in practice.
We propose to use a
simple framework (the augmented map) to model the problem.  Based on
the augmented map interface, we develop both multi-level tree
structures and sweepline algorithms supporting range, segment and rectangle
queries in two dimensions.
For the sweepline algorithms, we also propose a parallel paradigm
and show corresponding cost bounds. All of our data structures are work-efficient to build in
theory ($O(n\log n)$ sequential work) and achieve a low parallel depth
(polylogarithmic for the multi-level tree structures, and
$O(n^{\epsilon})$ for sweepline algorithms).
The query time is almost linear to the output size.

We have implemented all the data structures described in the paper
using a parallel augmented map library. Based on the library each data
structure only requires about 100 lines of C++ code.  We test their
performance on large data sets (up to $10^8$ elements) and a machine
with 72-cores (144 hyperthreads). The parallel construction achieves
32-68x speedup.
Speedup numbers on queries are up to 126-fold.
Our sequential implementation outperforms the CGAL library by at least 2x
in both construction and queries. Our sequential implementation can be slightly
slower than the R-tree in the Boost library in some cases (0.6-2.5x), but has significantly better query performance (1.6-1400x)
than Boost.

\end{abstract}

%% file: intro.tex
\vspace{-.1in}
\section{Introduction}
\label{sec:intro}
The range, segment and rectangle query problems are fundamental
problems in computational geometry, and also have extensive
applications in many
domains.
In this paper, we focus on the 2D Euclidean space.
The range query problem is
to maintain a set of points,
and to answer queries regarding a given rectangle.
The segment query problem is to maintain a set of non-intersecting
segments, and
to answer questions regarding a given vertical line.
The rectangle query problem is to maintain a set of rectangles, and to answer questions
regarding a given point.
For all problems, we discuss queries of both listing all queried elements (the \emph{list-all} query),
and returning the count of queried
elements (the \emph{counting} query).
Some other queries, like the weighted sum of all queried elements,
can be done by a variant of the counting queries.
Although we only discuss these three problems,
many other problems, such as rectangle-rectangle intersection queries,
can be solved by a combination of these three queries.
Efficient solutions to these
problems are mostly based on some variant of a
range-tree~\cite{bentley1979data}, a
segment-tree~\cite{BentleyWood80}, a sweep-line
algorithm~\cite{Shamos76}, or combinations of them.

In addition to the large body of work on sequential algorithms,
there has also been research on parallel structures for such
queries~\cite{aggarwal1988parallel,Atallah89,goodrich1996sweep,AtallahG86}.
However, efficient implementations of these structures can be
complicated. We know of few theoretically efficient 
implementations in parallel for range and segment query structures.  This
challenge is not only because of the delicate design of many
algorithmic details, but also because very different forms of data
structures and optimizations are required due to varied
combinations of input settings and query types.  The parallel
implementations we know
of~\cite{Ho97,basin2017kiwi,chan1999hierarchical,parallelrtree} do not have useful
theoretical bounds.

Our goal is to develop theoretically efficient algorithms which can be
implemented with ease and also run fast in practice, especially \emph{in parallel}.
To do this, we use a simple framework 
called \emph{augmented maps}~\cite{pam}.
The augmented map is an
abstract data type (ADT) based on ordered maps of key-value pairs, augmented
with an abstract ``sum'' called the \emph{augmented value} (see Section \ref{sec:prelim}).
The original paper \cite{pam} proposes to support augmented maps using augmented trees,
and presents a parallel augmented map library PAM,
in which all functions have asymptotical optimal sequential work and polylogarithmic parallel depth.
We use this library in our experiments. 
In this paper, as an alternative to
augmented trees, we also propose the \emph{\sweepstructure{}}
to support augmented maps, which
stores the augmented values of all the prefixes of
the original map. When using the \sweepstructure{}s to represent the
outer maps for the geometry queries, the corresponding algorithm resembles
the standard sweepline algorithm.  We propose a parallel paradigm to
construct the \sweepstructure{}s, and when the input functions
satisfy certain conditions, we show the corresponding cost bounds.

\hide{Previous work~\cite{pam} shows that using augmented
trees as the underlying structure, a wide variety of functions on
(augmented) maps can be implemented efficiently, including union,
insertion, range extraction and sums, and so
on.  All the bulk operations such as union and construction can run in
parallel with asymptotically optimal work and polylogarithmic depth.}

The augmented map framework provides concise and elegant abstraction
for many geometric queries, and is extendable to a wide range of problems and query types.
In particular, we study how it can model
the range, segment and rectangle query problems, all of which are
two-level map structures: an outer level map augmented with an inner
map structure. We use augmented trees to
implement the inner maps.  Using different representations as the outer
map, we develop both algorithms based on multi-level
augmented trees, and sweepline algorithms based on the \sweepstructure{s}.  

In this paper we present ten data
structures for range, segment and rectangle queries.   Five of them are
multi-level trees including the \emph{RangeTree} (for range query), the \emph{SegTree} (for segment query),
\emph{RecTree} (for rectangle query), and another two for fast counting queries \emph{SegTree*} (segment counting) and
\emph{RecTree*} (rectangle counting). The other five are corresponding sweepline algorithms.
All these data structures are efficient both in theory and practice.
Theoretical bounds of our implementations are summarized in Table~\ref{tab:allexpcost}.
Some of our bounds are not optimal, but we note that all of them are off
optimal by at most a $\log n$ factor in work.
Some of our algorithms are
motivated by previous theoretically efficient parallel
algorithms~\cite{aggarwal1988parallel,AtallahG86}.

\hide{For the range query problem and segment query problem, we first model each application
as a two-level augmented map structure. We always use the augmented tree to implement the inner map.
We discuss two implementations of the outer map: the augmented tree structure, which makes the whole data structure
a \emph{range tree} (noted as RangeTree) or a \emph{segment tree} (noted as SegTree), and the \sweepstructure{}, which makes the whole algorithm a \emph{sweepline}
algorithm (noted as RangeSweep, SegSweep and CountSweep).
We show that all of them fit the abstraction, which allows both the description and the implementation to be concise and
easy to understand. They are also efficient both in theory and in practice assuming \emph{persistence} \cite{driscoll1986making} (the input is not affected by modifications,
and updates always yield new versions) is supported. A list of all data structures and their corresponding costs are summarized in Table~\ref{tab:allexpcost}.}

Using the augmented map abstraction greatly simplifies engineering and
reduces the coding effort as can be indicated by the required lines of code---on top
of the PAM library, each application only requires about 100 lines of C++ code even for the
parallel version. We show some example codes in Appendix \ref{app:codeexample}.
Beyond being concise, the implementation also achieves good performance. We present experimental results of our implementations and also compare them to C++ libraries CGAL\cite{CGAL} and Boost\cite{boost}, which are both hundreds lines of code for sequential algorithms.
We get a 33-to-68-fold self-speedup in construction on a
machine with 72 cores (144 hyperthreads), and 60-to-126-fold speedup
in queries. Our sequential construction outperforms CGAL by more than a factor of 2, and is comparable to Boost. Our query time outperforms both CGAL and Boost by a factor of 1.6-1400.

\begin{table}
  \centering\small
\begin{tabular}{@{ }c@{ }|@{ }c@{ }|@{ }c@{ }|@{ }c@{ }|@{ }c@{ }|@{ }c@{ }}
\hline
 \multicolumn{1}{c@{ }}{} & & \multicolumn{2}{@{ }c|@{ }}{\textbf{Build}} & \multicolumn{2}{@{ }c}{\textbf{Query}} \\
\cline{3-6}
 \multicolumn{1}{c}{} & &       \textbf{Work} &      \textbf{Depth} &  \textbf{All} &  \textbf{Count} \\
\hline
     \textbf{Range} & \textbf{Swp.} & $n\log n$ & $n^{\epsilon}$ & $\log n+ k\log \left(\frac{n}{k}+1\right)$ & $\log n$ \\
     \textbf{Query} &  \textbf{Tree} & $n\log n$ & $\log^3 n$ & $\log^2 n+k$ & $\log^2 n$ \\
\hline
\textbf{Seg} &   \textbf{Swp.} & $n\log n$ & $n^{\epsilon}$  & $\log n+k$ & $\log n$ \\
\textbf{Query} &    \textbf{Tree} & $n\log n$ & $\log^3 n$ & $\log^2 n + k$ & $\log^2 n$ \\
\hline
\textbf{Rec} &   \textbf{Swp.} & $n\log n$ & $n^{\epsilon}$ & $\log n+k\log \left(\frac{n}{k}+1\right)$ & $\log n$ \\
\textbf{Query} &   \textbf{Tree} & $n\log n$ & $\log^3 n$ &   $\log^2 n+k\log \left(\frac{n}{k}+1\right)$ & $\log^2 n$ \\
\hline
\end{tabular}
\vspace{-.1in}
  \caption{\textbf{Theoretical (asymptotical) costs of all problems in this paper} - $n$ is the input size, $k$ the output size. $\epsilon<1$ can be any given constant. Bounds are in big-O notation. ``Swp.'' means the sweepline algorithms. ``Tree'' means the two-level trees. We note that not all queries have optimal work, but they are off optimal by at most a $\log n$ factor.}
  \label{tab:allexpcost}\vspace{-.3in}
\end{table}

%% file: related.tex
\section{Related Work}
\label{sec:related}
Many data structures are designed for solving range, segment and rectangle queries such as range trees~\cite{bentley1979decomposable}, segment trees\cite{bentley1980optimal}, kd-trees~\cite{kdtree}, R-trees~\cite{rstartree,packedrtree,seeger1991multi}, priority trees\cite{prioritytree}, and many others~\cite{willard1985new,CLRS,de2000computational,overmars1988geometric}, which are then applied to later research in various areas ~\cite{Agarwal16,jacob2017,ahle2017parameter,bille2015compressed,goodrich2010priority,brisaboa2016aggregated,sitchinava2012computational,ajwani2010geometric,arge2003optimal,arge1999two,bentley1979data,pagel1993towards,blankenagel1994external}.
Many of them are augmented tree structures.
The standard range tree has construction time $O(n\log n)$ and query time $O(k+\log^2 n)$ for input size $n$ and output size $k$. Using fractional cascading \cite{lueker1978data,willard1978predicate}, the query time can be reduced to $O(k+\log n)$. We did not employ such optimizations, but instead show that the our version using parallel augmented maps achieve good parallelism in practice, and is simple and easy for engineering. The terminology ``segment tree'' refers to different data structures in the literature. Our version is similar to some previous works \cite{Chaselle1984,AtallahG86,aggarwal1988parallel}. Previous solutions for rectangle queries usually use combinations of range trees, segment trees, interval trees, and priority trees~\cite{edelsbrunner1983new,cheng1990efficient,edelsbrunner1981intersection}. Their algorithms are all sequential.
Many previous results focus on developing fast sequential sweepline algorithms for range queries~\cite{arge2005cache,arge2006simple}, segment intersecting~\cite{mehlhorn1994implementation,chazelle1992optimal} and rectangle queries~\cite{mccreight1980efficient}.

In the parallel setting, there has also been a lot of theoretical works~\cite{aggarwal1988parallel,atallah1986efficient,Atallah89,goodrich1996sweep}.
Atallah et al.~\cite{Atallah89} propose cascading divide-and-conquer scheme for solving many computational geometry problems in parallel.
Goodrich et al.~\cite{goodrich1996sweep} propose a framework to parallelize several sweepline-based algorithms.
We know of no experimental evaluation of these algorithms.
There are also parallel implementation-based works such as parallel R-trees~\cite{parallelrtree}, parallel sweepline algorithms~\cite{mckenney2009parallel}, and algorithms focusing on distributed systems~\cite{zheng2006distributed} and GPUs~\cite{yu2011parallel}.
No theoretical guarantees are provided in these papers.
There also have been papers on I/O efficient computational geometry problems  \cite{sitchinava2012computational,ajwani2010geometric,arge2003optimal,arge1999two}.

%% file: prelim.tex
\section{Preliminaries}
\label{sec:prelim}

\para{Notation.} We call a key-value pair in a map an \emph{entry} denoted as $e=(k,v)$. We use $k(e)$ and $v(e)$ to extract the key and the value from an entry. Let $\langle P\rangle$ be a sequence of elements of type $P$. For a tree node $u$, we use $k(u)$, $\lc{u}$ and $\rc{u}$ to extract its key, left child and right child respectively.

On the 2D planar, let $X$, $Y$ and $D=X\times Y$ be the types of x- and y-coordinates and the type of points, where $X$ and $Y$ are two sets with total ordering defined by $<_X$ and $<_Y$ respectively.
For each point $p\in D$ on the 2D planar, we use $x(p)\in X$ and $y(p)\in Y$ to extract its x- and y-coordinates, and use a pair $(x(p), y(p))$ to denote $p$. For simplicity,
we assume all coordinates are unique. Duplicates can be resolved by slight variants of algorithms in this paper.

\para{Parallel Cost Model.} To analyze asymptotic theoretical
 costs of a parallel algorithm we use \emph{work} $W$ and \emph{depth} $D$ (or \emph{span} $S$), where work is the
total number of operations and depth is the length of the critical
path.  We can bound the total runtime in terms of work and depth since
any computation with $W$ work and $D$ depth will run in time
$T < \frac{W}{P} + D$ assuming a PRAM \cite{jaja1992introduction}
with $P$ processors and a greedy
scheduler~\cite{Graham69,Brent74,BL98}.  We assume concurrent reads and exclusive writes (CREW).

\para{Persistence.} A \emph{persistent} data structure \cite{persistent} is a data structure that preserves the previous version of itself when being modified and always yields a new updated structure. For BSTs, persistence can be achieved by path copying~\cite{Sarnak86}, in which only the affected path related to the update is copied, such that the asymptotical cost remains unchanged.
In this paper, we assume underlying persistent tree structures. In experiments, we use a library (PAM) supporting persistence~\cite{pam}.

\para{Sweepline Algorithms.} A sweepline algorithm (or plane sweep algorithm) is an algorithmic paradigm that uses a conceptual sweep line to process elements in order in Euclidean space~\cite{Shamos76}. It uses a virtual line sweeping across the plane, which stops at some points (e.g., the endpoints of segments) to make updates. We call the points the \emph{event points} $p_i\in P$. They are processed in a
total order defined by $\prec:P\times P\mapsto \bool$.
Here we limit ourselves to cases where the events are known ahead of time. As the algorithm processes the points, a data structure $T$ is maintained and updated at each event point to track the status at that point. Sarnak and Tarjan~\cite{Sarnak86} first noticed that by being persistent, one can keep the intermediate structures $t_i\in T$ at all event points, such that they can be used for later queries.
Indeed in the two sweepline algorithms in this paper, we adopt the same methodology, but parallelize it. We call them the \emph{\sweepstructure{s}} at each point.

Typically in sweepline algorithms, on encountering an event point $p_i$ we compute $t_i$ from the previous structure $t_{i-1}$ and the new point $p_i$ using an \emph{update function} $h:T\times P\mapsto T$, i.e., $t_i=h(t_{i-1}, p_i)$.  The initial structure is $t_0$.   A sweepline algorithm can
therefore be defined as the five tuple:
\begin{align}
\boldsymbol{S} = \text{\bf SW}(P,\prec,T,t_0,h)
\end{align}
\hide{
\begin{center}
\small
\begin{tabular}{|r@{ }c@{ }c@{ }@{ }c@{ }@{ }l@{ }l|}
  \hline
 $\bf S$&=& \textbf{SW}&(& $\boldsymbol{P}$; \quad $\boldsymbol{\prec}$: $P\times P\mapsto \bool$; \quad $\boldsymbol{T}$; \quad $\boldsymbol{t_0\in T}$; &\\
&&&&$\boldsymbol{h}:P\times T \mapsto T$ &)\\
  \hline
\end{tabular}
\end{center}
}
It defines a function that takes a set of points $p_i\in P$ and returns a mapping from each point to a \sweepstructure{} $t_i\in T$.

%% file: aug_map.tex
\begin{table}
\footnotesize\centering
\begin{tabular}{>{\bf}l<{}@{}|@{ }c@{ }|@{ }c@{ }}
\hline
\textit{\textbf{Function}} &\textit{\textbf{Work}}& \textit{\textbf{Depth}}\\
\hline
\mb{Insert$(m,e,\sigma)$}, \mb{Delete$(m,k)$}&$\log n$&$\log n$\\
\hline
\mb{Intersect}$(m_1,m_2)$& \multirow{3}{*}{$n_1\log \left(\frac{n_1}{n_2}+1\right)$}&\multirow{3}{*}{$\log n_1\log n_2$}\\
\mb{Difference}$(m_1,m_2)$ & &\\
\mb{Union}$(m_1, m_2, \sigma)$& &\\
\hline
\mb{Build}$(s, \sigma)$ & $n\log n$&$\log n$\\
\hline
\mb{UpTo}$(m,k)$, \mb{Range}$(m,k_1,k_2)$ & $\log n$&$\log n$\\
\hline
\mb{ALeft}$(m,k)$, \mb{ARange}$(m,k_1,k_2)$  &$\log n$&$\log n$\\
\hline
\end{tabular}
\caption{\textbf{Core functions on augmented map interface} - $k,k_1,k_2\in K$. $m, m_1, m_2$ are maps, $n=|m|, n_i=|m_i|$. $s$ is a sequence.  All bounds are in big-O notation. The bounds assume the augmenting functions $f$, $g$ and argument function $\sigma$ all have constant cost.}
    \label{tab:mapfunctions}
    \vspace{-.2in}
\end{table}
\para{The Augmented Map.}
The \emph{augmented map} \cite{pam} is an abstract data type (ADT) that associates an ordered \emph{map} (a set of key-value pairs) with a ``map-reduce'' operation for keeping track of the abstract sum over entries (referred to as the \emph{augmented value} of the map). More formally, an augmented map is an ordered map $M$ where keys belong to some ordered set $K$ (with total ordering defined by relation $<_K$) and values to a set $V$, that is associated with two functions: the \emph{base function} $g:K\times V \mapsto A$ that maps a key-value pair to an augmented value (from a set $A$) and the \emph{combine function} $f: A\times A \mapsto A$ that combines (reduces) two augmented values.  We require $f$ to be associative and have an identity $I$ (i.e., set $A$ with $f$ and $I$ forms a monoid).
An augmented map can therefore be defined as the seven tuple:
\begin{align}
a = \text{\bf AM}(K,<_K,V,A,g,f,I)
\end{align}


\hide{From the total order of $K$ and the associativity of $f$, we know the augmented value of a set is unique.
Specifically, note that $f$ has identity $a_{\emptyset}$, we know the augmented value on an empty set $\mathcal{A}(\emptyset)=a_{\emptyset}$.}

\hide{
\mapstructure{
{\small
\centering
\begin{tabular}{|@{}r@{ }c@{ }c@{ }c@{ }l@{ }l|}
  \hline
  \textbf{AugMap} &=& \textbf{AM}&(& $\boldsymbol{K}$;\, $\boldsymbol{\prec}: K\times K\mapsto\bool$;\, $\boldsymbol{V}$;& \\
  &&&&$\boldsymbol{A}$; \,$\boldsymbol{g}: K\times V \mapsto A$; \,$\boldsymbol{f}:A\times A \mapsto A$; \, $\boldsymbol{I}\in A$ &)\\
  \hline
\end{tabular}
}}}

Then the augmented value of a map $M=((k_1,v_1), \dots, (k_n,v_n))$, denoted as $\mathcal{A}(M)$, is defined as $\mathcal{A}(M) = f(g(k_1, v_1), g(k_2, v_2), \dots, g(k_n, v_n))$, where the definition of the binary function $f$ is extended as:
\begin{align}
\nonumber
f(\emptyset) &= I, f(a_1) = a_1,\\
\nonumber
f(a_1, a_2, \dots, a_n) &= f(f(a_1, a_2, \dots a_{n-1}), a_n)
\end{align}

A list of common functions on augmented maps used in this paper, and their parameters, are shown in Table \ref{tab:mapfunctions} (the first column).
Functions related to augmentations that are useful in this paper include: the \mb{ARange} function which returns the augmented value of all entries within a key range, and \mb{ALeft}$(k)$ which returns the augmented value of all entries up to a key $k$. More details are in \cite{pam} and Appendix \ref{app:pam}.

\hide{The abstraction we propose in this paper does not rely on the underlying implementation of augmented maps, nor any specific balancing schemes. However some bounds and proofs may be only applicable on specific data structures (e.g., the insertion on range trees requires a weight-balanced tree to guarantee the theoretical bound, and an efficient query on the segment tree especially makes use of the tree structure).
In the experiments we use a library PAM \cite{pam} using weight-balanced trees implementing augmented maps (See details in the appendix).}

\para{Augmented Maps Using Augmented Trees.}
An efficient implementation of augmented maps is to use augmented balanced binary search trees~\cite{pam}. Entries are stored in tree nodes and sorted by keys. Each node also maintains the \emph{augmented value} of the subtree rooted at it. Using join-based algorithms~\cite{join,pam} on trees, the augmented map interface can be supported in an efficient and highly-parallelized manner, and the costs are listed in Table \ref{tab:mapfunctions}. All functions listed in Table \ref{tab:mapfunctions} have optimal work and polylog depth. In the experiment we use a parallel library PAM \cite{pam} which implements augmented maps using augmented trees. The cost of functions in the library matches the bounds in Table \ref{tab:mapfunctions}.

%% file: sweepline.tex
\section{Augmented Maps Using Prefix Structures}
\label{sec:augsweep}
In this paper, as an alternative to the tree structure, we propose to use the \emph{\sweepstructure{s}} as in the sweepline algorithms to represent augmented maps. We especially use \sweepstructure{s} to represent the outer map structure in range, segment and rectangle queries, which makes the algorithm equivalent to a sweepline algorithm.
For an augmented map $m=\{e_1,\dots, e_{|m|}\}$, \sweepstructure{s} store the augmented values of all prefixes up to each entry $e_i$, i.e., $\mb{ALeft}(m,k(e_i))$. For example, if the augmented value is the sum of values, the \sweepstructure{s} are prefix sums.
This is equivalent to using a combination of function $f$ and $g$ as the update function. That is to say, an augmented map $m=\text{\bf AM}(K,\prec,V,A,f,g,I)$ is equivalent to (or, can be represented by) a sweepline scheme $S$ as:
\begin{align}
\label{eqn:sweep}
S=\text{\bf SW}(K\times V, \prec, A,t_0\equiv I,h(t,p)\equiv f(t,g(p)))
\end{align}
In many cases a much simpler update function $h(t,p)$ can be provided as a replacement for $f(t,g(p)))$.

\para{A Parallel Sweepline Paradigm.}
Here we present a parallel algorithm to build the \sweepstructure{}s. Because of the associativity of the combine function, to repeatedly update a sequence of points $\langle p_i\rangle$ onto a \sweepstructure{} $t$ using $h$ is equivalent to directly combining the augmented value on all points $\langle p_i\rangle$ to $t$ using the combine function $f$.
Thus our approach is to evenly split the input sequence of points into $b$ blocks, calculate the augmented value of each block, and then refine the \sweepstructure{}s in each block using the update function $h$. For simplicity we assume $n$ is an integral multiple of $b$ and $n=b\times r$. We define a \emph{fold function} $\funca :\langle P \rangle \mapsto T$ that converts a sequence of points into the augmented value.  
The parallel sweepline paradigm can therefore be defined as the six tuple:
\begin{align}
 \label{eqn:ps}
\bf S'=\text{\bf PS}(&\boldsymbol{P}; \boldsymbol{\prec}; \boldsymbol{T}; \boldsymbol{t_0};\boldsymbol{h}; \boldsymbol{\funca}; \boldsymbol{f})
\end{align}
The algorithm to build the \sweepstructure{}s is as follows (also see Algorithm \ref{algo:sl}):

\begin{algorithm}[!t]
  \KwIn{A list $p$ of length $n$ storing all input points in order, the update function $h$, the fold function $\funca$, the combine function $f$, the empty \sweepstructure{} $t_0=I$ and the number of blocks $b$. We assume $n=br$.
    }
\KwOut{A series of \rangesweep{} $t_i$.}
    \vspace{0.5em}
    \DontPrintSemicolon
    \begin{minipage}{\columnwidth}
    \SetKwFor{myforone}{[Step 1] parallel\_for}{do}{endfor}
    \SetKwFor{myfortwo}{[Step 2] for}{do}{endfor}
    \SetKwFor{myforthree}{[Step 3] parallel\_for}{do}{endfor}
    \myforone {$i\leftarrow 0$ to $b-1$} { \label{line:rwstep111}
        $t'_i=\funca(p_{i\times r}, \dots, p_{(i+1)\times r-1})$ \label{line:rwstep1e}
    }
    \myfortwo {$i\leftarrow 1$ to $b-1$} {\label{line:rwstep222}
        $t_{i\times r} = f(t_{(i-1)\times r}, t'_{i-1})$ \label{line:rwstep2ee}
    }
    \myforthree {$i\leftarrow 0$ to $b-1$} { \label{line:rwstep333}
        $s = i\times r$; $e = s+r-1$\\
        \lFor {$j\leftarrow s$ to $e$} {
            $t_{j} = h(t_{j-1}, p_j)$
        }
        \label{line:rwstep3ee}
    }
    \Return {$\{ p_i \mapsto t_i\}$}
    \end{minipage}\hfill
    \begin{minipage}{\columnwidth}
\begin{flushright}
   \includegraphics[width=\columnwidth]{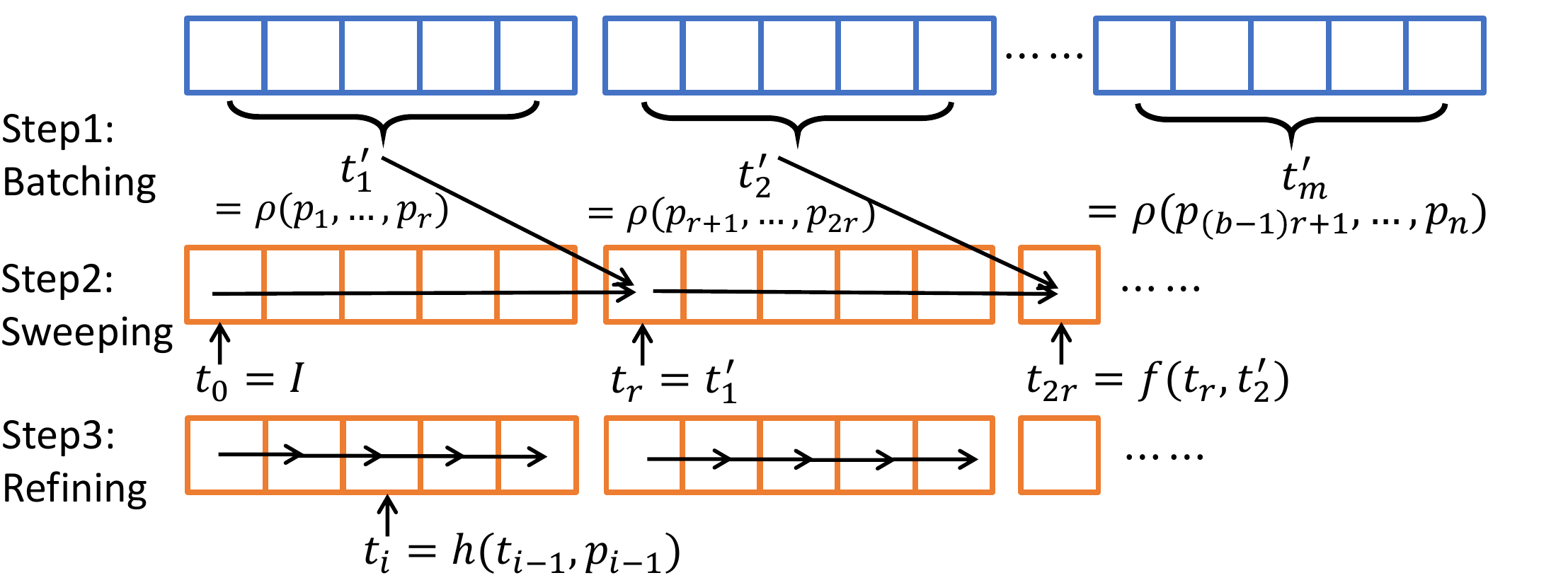}\\
\end{flushright}
    \end{minipage}
\caption{The construction of the \sweepstructure{}.}\label{algo:sl}
\end{algorithm}

\hide{
\begin{SCfigure}
  \includegraphics[width=0.6\columnwidth]{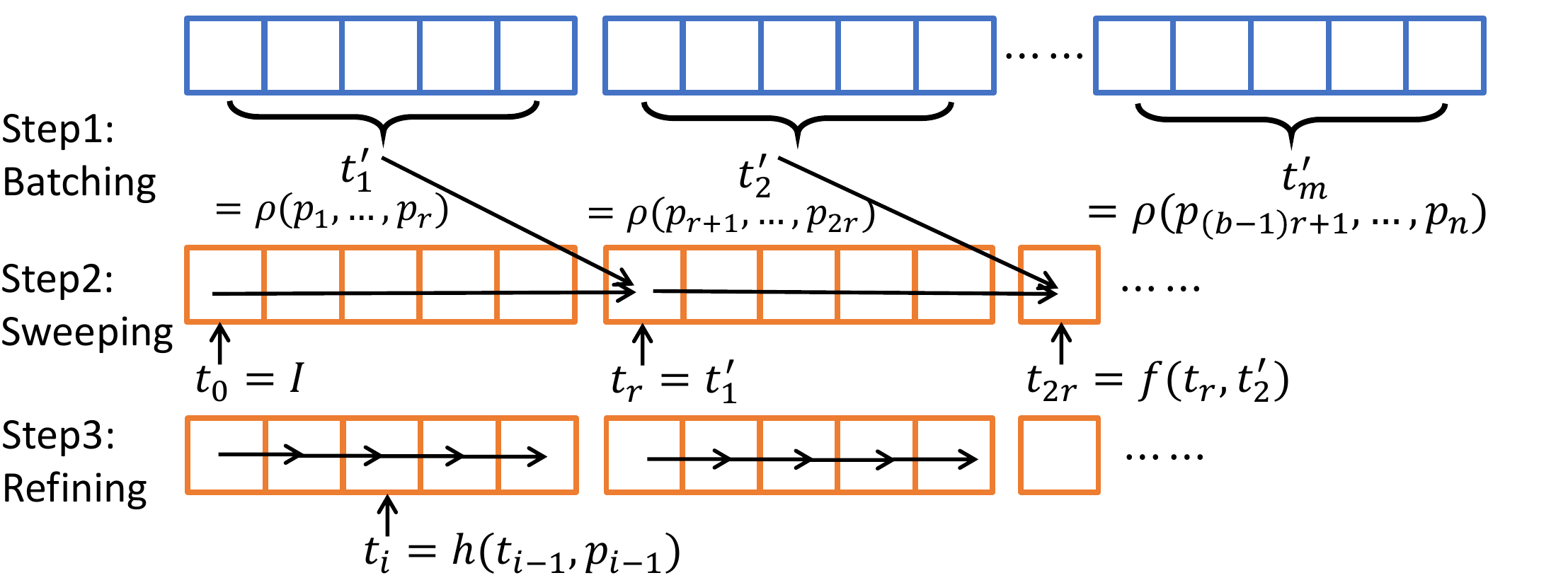}\\
  \caption{An illustration for the parallel sweepline paradigm. The blue boxes denotes the event points. The orange boxes are the \sweepstructure{}s. }\label{fig:sweepline}
\end{SCfigure}}

\begin{enumerate}[noitemsep,topsep=0pt]
\item \textbf{Batch}. Assume all input points have been sorted by $\prec$. We evenly split them into $b$ blocks and then in parallel generate $b$ augmented values (partial sums) $t'_{i}\in T$ using $\funca$, each corresponding to one of the $b$ blocks. 

\item \textbf{Sweep}. These partial sums $t'_i$ are combined in turn sequentially by the combine function $f$ to get the first \sweepstructure{} $t_0,t_{r},t_{2r}\dots$ in each block using $t_{i\times r} = f(t_{(i-1)\times r}, t'_i)$.

\item \textbf{Refine}. All the other \sweepstructure{}s are built based on $t_0, t_r,t_{2r}\dots$ (built in the second step) in the corresponding blocks. All the $b$ blocks can be done in parallel. In each block, the points $p_i$ are processed in turn to update its previous \sweepstructure{} $t_{i-1}$ sequentially using $h$.
\end{enumerate}

Here $\prec:P\times P\mapsto \bool$, $t_0\in T$, $h:T\times P\mapsto T$, $\funca:\langle P\rangle \mapsto T$ and $f:T\times T\mapsto T$ are as defined above. In many non-trivial instantiations of this framework (especially those in this paper), each \sweepstructure{} keeps an ordered set tracking some elements related to a subset of the processed event points, having size $O(i)$ at point $i$. The combine function is some set function (e.g.,\mb{Union}), which typically requires $O(n_2\log (n_1/n_2+1))$ work for combining a block of size $n_2$ with the current \sweepstructure{} of size $n_1\ge n_2$~\cite{join}. Accordingly, the function $h$ simply updates one element (e.g., an insertion corresponding to a \mb{Union}) to the structure, costing $O(\log n)$ on a structure of size $n$.
Creating the augmented value of each block of $r$ points (using $\funca$) costs $O(r\log r)$ work and $O(\log r)$ depth, building a \sweepstructure{} of size at most $r$. A common setting of the related functions is summarized in Table \ref{tab:sweepcost}, and the corresponding bounds of the sweepline algorithm is given in Theorem \ref{thm:sweeptheory}.
\begin{table}
  \centering
  \small{
\begin{tabular}{@{}c@{}|@{ }c@{ }|@{ }c@{ }|@{ }c@{}}
\hline
&$\boldsymbol{h(t,p)}$ & $\boldsymbol{\funca(s)}$ & $\boldsymbol{f(t,t')}$ \\
\hline
\textbf{Work} & $O(\log |t|)$ & $O(|s|\log |s|)$ & $O(|t'|\log (|t|/|t'|+1))$\\
\hline
\textbf{Depth} & $O(\log |t|)$ & $O(\log |s|)$ & $O(\log |t|\log |t'|)$\\
\hline
\textbf{Output} & $O(|t|)$ & $O(|s|)$ & $O(|t'|+|t|)$\\
\hline
\end{tabular} }\vspace{-.1in}
  \caption{\textbf{A typical setting of the function costs in a sweepline paradigm.}}\label{tab:sweepcost}\vspace{-.2in}
\end{table}

\begin{theorem}
\label{thm:sweeptheory}
A sweepline paradigm $S$ as in Equation \ref{eqn:sweep} can be built in parallel using its corresponding parallel paradigm $S'$ (Equation \ref{eqn:ps}).  If the bounds as shown in Table \ref{tab:sweepcost} hold, then Algorithm \ref{algo:sl} can construct all \sweepstructure{}s in work $O(n\log n)$ and depth $O(\sqrt{n}\log^{1.5}n)$.
\end{theorem}
\begin{proof}
\hide{
\vspace{-.5em}
\mapstructure{
{\small
\centering
\begin{tabular}{|r@{ }@{ }c@{ }@{ }c@{ }@{ }c@{ }@{ }lllllll@{ }l|}
  \hline
  $\boldsymbol{A_S}$ &=& \textbf{PS}&(& $\boldsymbol{P}$: $P_0$; & $\boldsymbol{\prec}$: $<_P$; &$\boldsymbol{T}$: $T_0$; & $\boldsymbol{h}$: insert or delete &&&&)\\
  $\Rightarrow\boldsymbol{A_P}$ &=& \textbf{PP}&(& $\boldsymbol{P}$: $P_0$; & $\boldsymbol{\prec}$: $<_P$; &$\boldsymbol{T}$: $T_0$; & $\boldsymbol{h}$: insert or delete&$\boldsymbol{T'}$: $\langle T_0, T_0\rangle$; & $\boldsymbol{\funca}$: $\rho_{0}$; &$\boldsymbol{\funcb}$: $(\langle t_I,t_D\rangle,t)\mapsto t\cup t_I \backslash t_D$ &)\\
  \multicolumn{1}{|r}{, where}&\multicolumn{11}{l|}{$T_0$ is a set of element $E$, $\boldsymbol{\funca_{0}}: \{p_i\}\mapsto \langle t_I, t_D\rangle$. $t_I$ stores all $e\in E$ to be inserted related to points in sequence $\{p_i\}$,} \\
  &&&&&&\multicolumn{5}{l}{$t_D$ for those to be deleted.}&\\
  \hline
\end{tabular}
}}}
The algorithm and its cost is analyzed as follows:
\begin{enumerate}[noitemsep,topsep=1pt]
\item \textbf{Batch}. Build $b$ units of $t'_i\in T$ using function $\funca$ in each block in parallel. There are $b$ such structures, each of size at most $n/b$, so it takes work $O(b\cdot \frac{n}{b}\log \frac{n}{b})=O(n\log \frac{n}{b})$ and depth $O(\log \frac{n}{b})$.

\item \textbf{Sweep}. Compute $t_r,t_{2r}\dots$ by combining $t'_i$ of each block with the previous \sweepstructure{} using the combine function $f$, sequentially. The calculation of each \sweepstructure{} is sequential, but the combine function works in parallel. The size of $t'_i$ is no more than $O(n/b)$. Considering the given work and depth bounds of the combine function, the total work of this step is bounded by: $O\left(\sum_{i=1}^{b} r\log (\frac{ir}{r} +1)\right)=O(n\log b)$. The depth is:
$O\left(\sum_{i=1}^{b}\log r \log ir\right)=O(b\log \frac{n}{b}\log n)$.
\item \textbf{Refine}. Build all the other \sweepstructure{}s using $h$. The total work is: $O\left(\sum_{i=1}^{n} \log i\right)=O(n\log n)$. We process each block in parallel, so the depth is $O(\frac{n}{b}\log n)$.
\end{enumerate}

In total, it costs work $O(n\log n)$ and depth $O\left(\left(b\log \frac{n}{b}+\frac{n}{b}\right)\log n\right)$. When $b=\Theta(\sqrt{n/\log n})$, the depth is $O(\sqrt{n}\log^{1.5}n)$.
\end{proof}

By repeatedly applying this process to each block in the last step, we can further reduce the depth.

\begin{corollary}
\label{coro:depth}
A sweepline paradigm $S$ as in Equation \ref{eqn:sweep} can be parallelized using its corresponding parallel paradigm $S'$ (Equation \ref{eqn:ps}).  If the bounds as shown in Table \ref{tab:sweepcost} hold, then we can construct all \sweepstructure{s} in work $O(\frac{1}{\epsilon}n\log n)$ and depth $\tilde{O}(n^{\epsilon})$ for arbitrary small $\epsilon>0$.
\end{corollary}
We give the proof in Appendix \ref{app:rangesweepdepth}.


In this paper, we use \sweepstructure{s} to represent the outer maps for the range, segment and rectangle queries. They all fit the parallel paradigm in Algorithm \ref{algo:sl}, and accord with the assumption on the function cost in Theorem \ref{thm:sweeptheory}. It is also easy to implement such a parallel algorithm. We show the code in Appendix \ref{app:sweepcode}, which is no more than half a page.

%% file: application.tex
\begin{table*}[!t]
{\small
\begin{tabular}{r@{}r@{ }@{}c@{ }@{}c@{ }@{ }c@{ }@{}l@{ }@{}l@{ }@{}l@{ }@{}l@{ }@{}l@{}@{ }l@{}@{ }l@{}@{}l@{}}
\hline
\multicolumn{13}{@{}l}{\normalsize{\textbf{* Range Query:}}}\\
\multicolumn{2}{@{}r@{}}{\bf(Inner Map)$\boldsymbol{R_I}$} &=& \textbf{AM} &(& $\boldsymbol{K}$: $D$;& $\boldsymbol{\prec}$: $<_{Y}$;&$\boldsymbol{V}$: $\mathbb{Z}$; &$\boldsymbol{A}$: $\mathbb{Z}$;& $\boldsymbol{g}$: $(k,v)\mapsto 1$; &$\boldsymbol{f}$: $+_{\mathbb{Z}}$; &$\boldsymbol{I}$: $0$ &)\\
\bf -&{\bf RangeMap $\boldsymbol{R_M}$} &=& \textbf{AM} &(& $\boldsymbol{K}$: $D$;& $\boldsymbol{\prec}$: $<_{X}$; &$\boldsymbol{V}$: $\mathbb{Z}$; &$\boldsymbol{A}$: $R_I$; & $\boldsymbol{g}$: $R_I$.singleton; &$\boldsymbol{f}$: $R_I$.union;& $\boldsymbol{I}$: $\emptyset$ &\multicolumn{1}{@{}l@{}}{)}\\
\bf -&{\bf RangeSwp $\boldsymbol{R_S}$}&=& \textbf{PS} &(& $\boldsymbol{P}$: $D$; &$\boldsymbol{\prec}$: $<_X$; & $\boldsymbol{T}$: $R_I$;&
   $\boldsymbol{t_0}$: $\emptyset$ &$\boldsymbol{h}$: $R_I$.insert& $\boldsymbol{\rho}$: $R_I$.build; & $f$: $R_I$.union &\multicolumn{1}{@{}l@{}}{)}\\
  \hline
\multicolumn{13}{@{}l}{\textbf{\normalsize{* Segment Query:}}}\\
\multicolumn{2}{@{}r@{}}{\bf(Inner Map)$\boldsymbol{S_I}$}&=& \textbf{AM}&(& $\boldsymbol{K}$: $D\times D$;& $\boldsymbol{\prec}$: $<_{Y}$; & $\boldsymbol{V}$: $\emptyset$;& $\boldsymbol{A}$: $\mathbb{Z}$; &$\boldsymbol{g}$: $(k,v)\mapsto 1$; &$\boldsymbol{f}$: $+_{\mathbb{Z}}$; &$\boldsymbol{I}$: $0$ &)\\
  \bf -&{\bf SegMap $\boldsymbol{S_M}$}&=& \textbf{AM}&(& $\boldsymbol{K}$: $X$;& $\boldsymbol{\prec}$: $<_{X}$; &$\boldsymbol{V}$: $D\times D$;&$\boldsymbol{A}$: $S_I\times S_I$; & $\boldsymbol{g}$: $g_{\text{seg}}$ &$\boldsymbol{f}$: $f_{\text{seg}}$&$\boldsymbol{I}$: $(\emptyset,\emptyset)$&\multicolumn{1}{@{}l@{}}{)}\\
\multicolumn{1}{l}{}&&\multicolumn{11}{@{}l@{}}{$\boldsymbol{g_{\text{seg}}(k, (p_l, p_r))}$: $\begin{cases}
        (\emptyset, S_I \text{.singleton}(p_l,p_r), \text{when } k = x(p_l)\\
        (S_I\text{.singleton}(p_l,p_r), \emptyset), \text{when } k = x(p_r)\\
        \end{cases}$, $\boldsymbol{f_{\text{seg}}}$: See Equation \ref{eqn:segcombine};} \\
  \bf -&{\bf SegSwp $\boldsymbol{S_S}$}&=& \textbf{PS}&(& $\boldsymbol{P}$: $D\times D$;&$\boldsymbol{\prec}$: $<_X$;&$\boldsymbol{T}$: $S_I$;& $\boldsymbol{t_0}$: $\emptyset$;&$\boldsymbol{h}$: $h_{\text{seg}}$;& $\boldsymbol{\funca}$: $\rho_{\text{seg}}$;& $\boldsymbol{f}$: $f_{\text{seg}}$ &\multicolumn{1}{@{}l@{}}{)}\\
  \multicolumn{13}{@{ }l@{}}{$\boldsymbol{h_{\text{seg}}(t,p)}=$ $\begin{cases}
        S_I.\text{insert}(t,p), \text{when~}p\text{~is a left endpoint} \\
        S_I.\text{delete}(t,p), \text{when~}p\text{~is a right endpoint} \\
        \end{cases}, \boldsymbol{\funca_{\text{seg}}}(\langle p_i\rangle)=\langle L,R\rangle\begin{cases}
        L\in S_I\text{: segments with \emph{right} endpoint in }\langle p_i\rangle\\
        R\in S_I\text{: segments with \emph{left} endpoint in }\langle p_i\rangle\\
        \end{cases}$}\\
        \hline
\multicolumn{13}{@{}l}{\textbf{\normalsize{* Rectangle Query:}}}\\
\multicolumn{2}{@{}r@{}}{\bf(Inner Map)$\boldsymbol{G_I}$}&=& \textbf{AM}&(& $\boldsymbol{K}$: $Y$;& $\boldsymbol{\prec}$: $<_{Y}$; & $\boldsymbol{V}$: $D\times D$;& $\boldsymbol{A}$: $Y$; &$\boldsymbol{g}$: $(k,(p_l,p_r))\mapsto y(p_r)$; &$\boldsymbol{f}$: $\max_Y$; &$\boldsymbol{I}$: $-\infty$ &)\\
  {\bf -}&\bf RecMap $\boldsymbol{G_M}$&=& \textbf{AM}&(& $\boldsymbol{K}$: $X$;& $\boldsymbol{\prec}$: $<_{X}$; &$\boldsymbol{V}$: $D\times D$;&$\boldsymbol{A}$: $G_I\times G_I$; & $\boldsymbol{g}$: $g_{\text{rec}}$ &$\boldsymbol{f}$: $f_{\text{rec}}$&$\boldsymbol{I}$: $(\emptyset,\emptyset)$&\multicolumn{1}{@{}l@{}}{)}\\
  {\bf -} &\bf RecSwp $\boldsymbol{G_S}$&=& \textbf{PS}&(& $\boldsymbol{P}$: $D\times D$;&$\boldsymbol{\prec}$: $<_X$;&$\boldsymbol{T}$: $G_I$;& $\boldsymbol{t_0}$: $\emptyset$;&$\boldsymbol{h}$: $h_{\text{rec}}$;& $\boldsymbol{\funca}$: $\rho_{\text{rec}}$;& $\boldsymbol{f}$: $f_{\text{seg}}$ &\multicolumn{1}{@{}l@{}}{)}\\
\multicolumn{1}{l}{}&&\multicolumn{11}{@{}l@{}}{$g_{\text{rec}}$, $f_{\text{rec}}$, $h_{\text{rec}}$ and $\rho_{\text{rec}}$ are defined similarly as $g_{\text{seg}}$, $f_{\text{seg}}$, $h_{\text{seg}}$ and $\rho_{\text{seg}}$} \\
       \hline
\end{tabular}
}
\caption{\textbf{Definitions of all structures in this paper} - Although this table seems complicated, we note that it fully defines all data structures used in this paper. $X$ and $Y$ are types of x- and y-coordinates. $D=X\times Y$ is the type of a point.}
\label{tab:allstructures}\vspace{-.2in}
\end{table*}

\section{2D Range Query}
\label{sec:range}

Given a set of $n$ points in the
plane, a \emph{range query} looks up
some information of points within a rectangle defined by a horizontal range
$(x_L, x_R)$ and vertical range $(y_L, y_R)$.

\hide{The 2D range query can be answered with a two-level map structure, as shown in Table \ref{tab:allstructures}.  The outer map ($M_R$) stores all the input points ordered by x-coordinates, and each node in this tree stores
an inner tree with all nodes in its
subtree ordered by y-coordinates. }

The 2D range query can be answered using a two-level map structure \emph{RangeMap}, each level corresponding to one dimension of the coordinates. It can answer both counting queries and list-all queries.
The definition (the outer map $R_M$ with inner map $R_I$) and an illustration are shown in Table \ref{tab:allstructures} and Figure \ref{fig:illustration} (a). In particular, the key of the outer map is the coordinate of each point and the value is the count. The augmented value of such an outer map, which is the inner map, contains the same set of points, but sorted by y-coordinates. Therefore, the base function of the outer map is just a singleton on the point and the combine function is \mb{Union}. The augmented value of the inner map counts the number of points in this (sub-)map.
Then the construction of a sequence $s$ of points can be done with the augmented map interface as: $r_M=R_M.\mb{Build}(s)$.

To answer queries, we use two nested range searches $(x_L,x_R)$ on the outer map and $(y_L,y_R)$ on the corresponding inner map. The counting query can be represented using the augmented map interface as:
\begin{align}
\label{eqn:rangequery}
\nonumber
\mb{RangeQuery}&(r_M,x_L,x_R,y_L,y_R) =\\
 R_I.\mb{ARange}&(R_M.\mb{ARange}(r_M,x_L,x_R), y_L, y_R)
\end{align}
The list-all query can be answered similarly using $R_I$.\mb{Range} instead of $R_I$.\mb{ARange}.

In this paper we use augmented trees for inner maps. We discuss two implementations of the outer map: the augmented tree, which makes the outer tree a range tree, and the \sweepstructure{}s, which makes the algorithm a sweepline algorithm.

\subsection{2D Range Tree}
\label{sec:rangetree}
If the outer map is supported using the augmented tree structure, the \emph{RangeMap} becomes a range tree (\emph{RangeTree}). In this case we do not explicitly build $R_M.\mb{ARange}(r_M,x_L,x_R)$ in queries. Instead, as the standard range tree query algorithm, we search the x-range on the outer tree, and conduct the y-range queries on the related inner trees. This operation is supported by the function \mb{AProject} in the augmented map interface and the PAM library. Such a tree structure can be constructed within work $O(n\log n)$ and depth $O(\log^3 n)$ (theoretically the depth can be easily reduced to $O(\log^2 n)$, but in the experiments we use the $O(\log^3 n)$ version to make fewer copies on data). It answers the counting query in $O(\log^2 n)$ time, and report all queried points in $O(k+\log^2 n)$ time for output size $k$. Similar discussion of range trees is shown in \cite{pam}. In this paper, we further discuss efficient updates on \emph{RangeTree} using the augmented map interface in Appendix \ref{app:wbfast}.

\subsection{The Sweepline Algorithm}
\label{sec:rangesweep}
In this section, we present a parallel sweepline algorithm \emph{RangeSwp} for 2D range query using our parallel sweepline paradigm, which can answer counting queries quickly. We use the \sweepstructure{s} to represent the outer map. Then each \sweepstructure{} is an inner map tracking all points up to the current point. The combine function of the outer map is \mb{Union}, so the update function $h$ can be an insertion. The corresponding fold function $\rho$ builds an inner map from a list of points. The definition of such a sweepline paradigm $R_S$ is shown in Table \ref{tab:allstructures}.

The theoretical bound of the functions (\mb{Insert}, \mb{Build}, and \mb{Union}) on the inner map, when supported by the augmented trees, are consistent with the assumptions in Theorem \ref{thm:sweeptheory}. Thus the theoretical cost of this algorithm directly follows from Theorem \ref{thm:sweeptheory}. Also, if persistence is supported by path-copying, this data structure takes $O(n\log n)$ space instead of trivially $O(n^2)$. Note that in previous work \cite{persistent} a more space-efficient version (linear space) is shown, but our point here is to show that our paradigm is generic and simple for many different problems without much extra cost.

\para{Answering Queries.} Computing $\mb{ARange}(r_M,x_L,x_R)$ explicitly in Equation \ref{eqn:rangequery} on \emph{RangeSwp} can be costly. We note that it can be computed by taking a \mb{Difference} on the \sweepstructure{} $t_R$ at $x_R$ and the \sweepstructure{} $t_L$ at $x_L$ (each can be found by a binary search). If only the count is required, a more efficient query can be applied.
We can compute the number of points in the range $(y_L,y_R)$ in $t_{L}$ and $t_{R}$ respectively, using \mb{ARange}, and the difference of them is the answer.
Two binary searches cost $O(\log n)$, and the range search on y-coordinate costs work $O(\log n)$. Thus the total cost of a single query is $O(\log n)$.

\para{Extension to Report All Points.} This sweepline algorithm can be inefficient in list-all queries. Here we propose a variant for list-all queries in Appendix \ref{app:rangesweeplist}.
The cost of one query is $O(\log n+ k\log (\frac{n}{k}+1))$ for output size $k$. Comparing with \emph{RangeTree}, which costs $O(k+\log^2 n)$ per query, this algorithm is asymptotically more efficient when $k<\log n$.

\hide{As mentioned above, the final answer is the difference of two such sets (subsets of $T_{x_l}$ and $T_{x_r}$). Since the complexity of the difference function depends on the size of both sets ($O(n\log (m/n+1))$), it is possible that the subset of $T_{x_l}$ is larger than the output size and the cost would be far higher than optimal. For example, when $x_l=O(n)$ and $x_r=x_l+c$ where $c$ is a constant, and the query range of y-coordinate is $(-\infty, +\infty)$. In this case reporting all points by taking the difference of $T_{x_l}$ and $T_{x_r}$ requires $O(n)$ time, but using a range tree would only take $O(\log^2 n)$ time (considering the output size is no more than $c$).}

\hide{In all, the augmented structure is defined as:
{\ttfamily\small
\begin{lstlisting}[language=C++,frame=lines,escapechar=@]
struct aug_max {
  using aug_t = int;
  static aug_t base(int y, int x){ return x;}
  static aug_t combine(int a, int b){ return max(a, b);}
  static aug_t get_empty() {return 0;}
};
\end{lstlisting}
}
}


%% file: app-seg.tex
\begin{figure*}[t]
\vspace{-.1in}
\begin{tabular}{@{}cc@{}}
\textbf{    (a) Range Query}&\textbf{    (b) Segment Query}\\
  \includegraphics[width=.48\columnwidth]{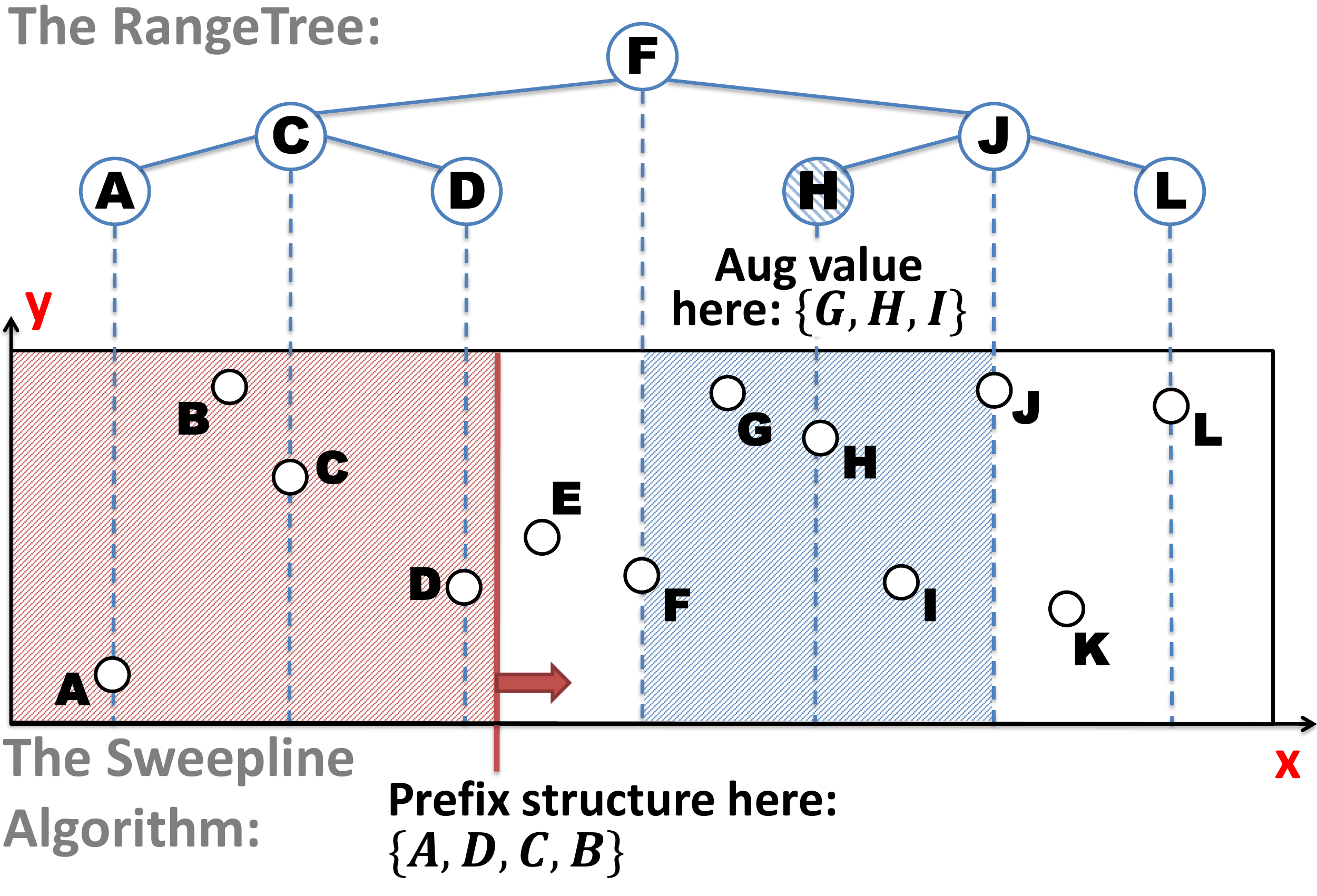}&
  \includegraphics[width=.48\columnwidth]{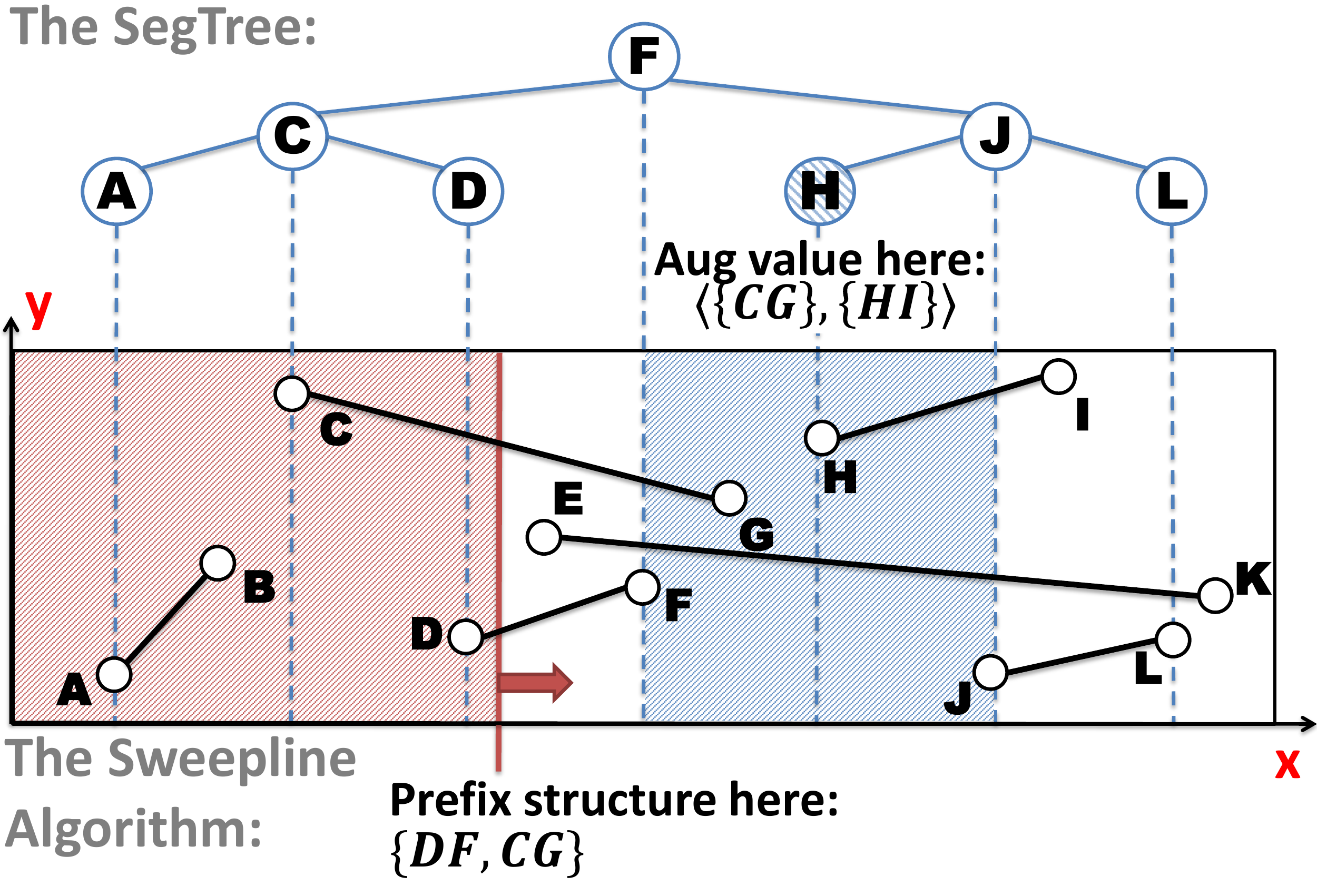}\\
\textbf{    (c) Rectangle Query}&\textbf{    (d) Segment Count Query}\\
  \includegraphics[width=.48\columnwidth]{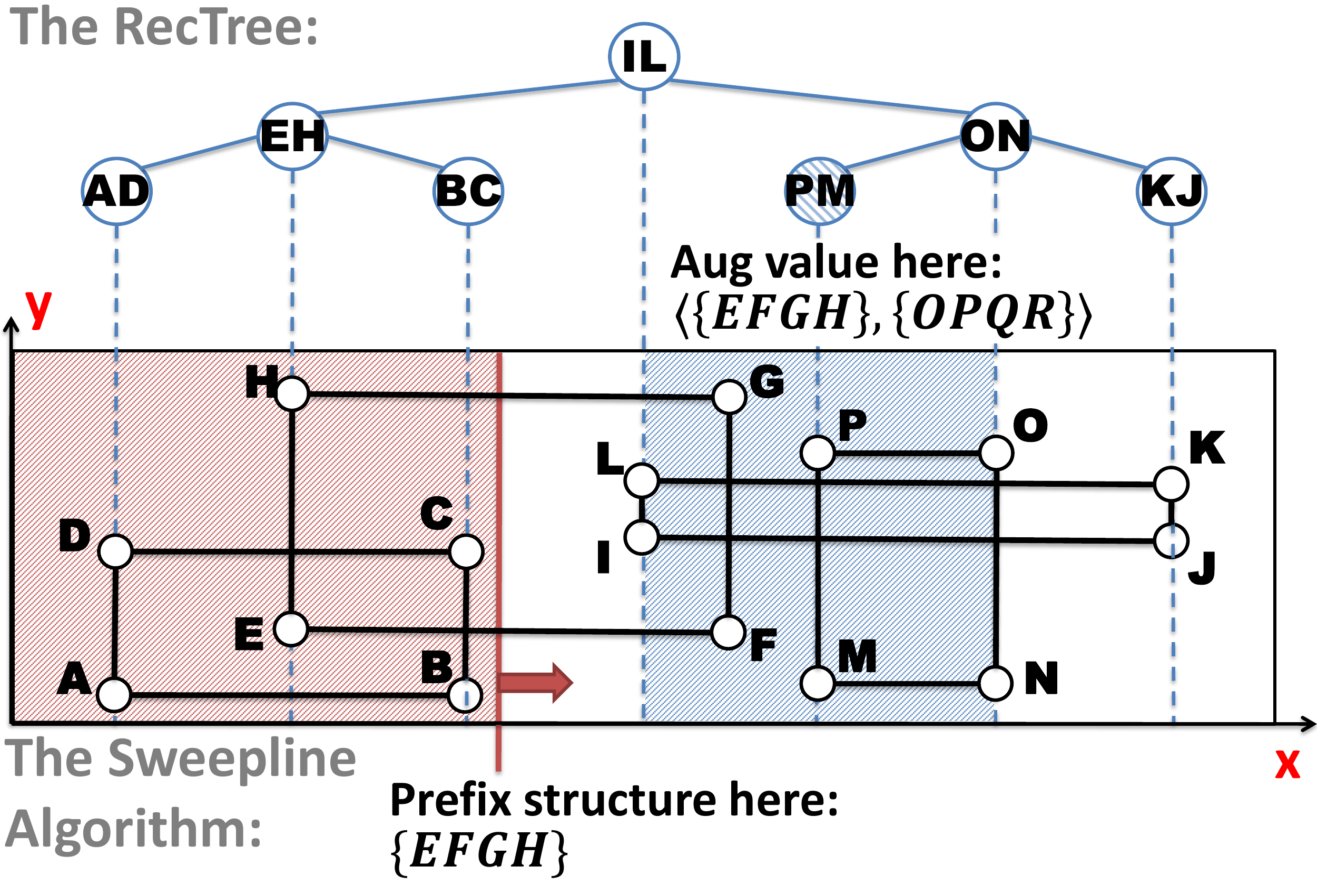}&
  \includegraphics[width=.48\columnwidth]{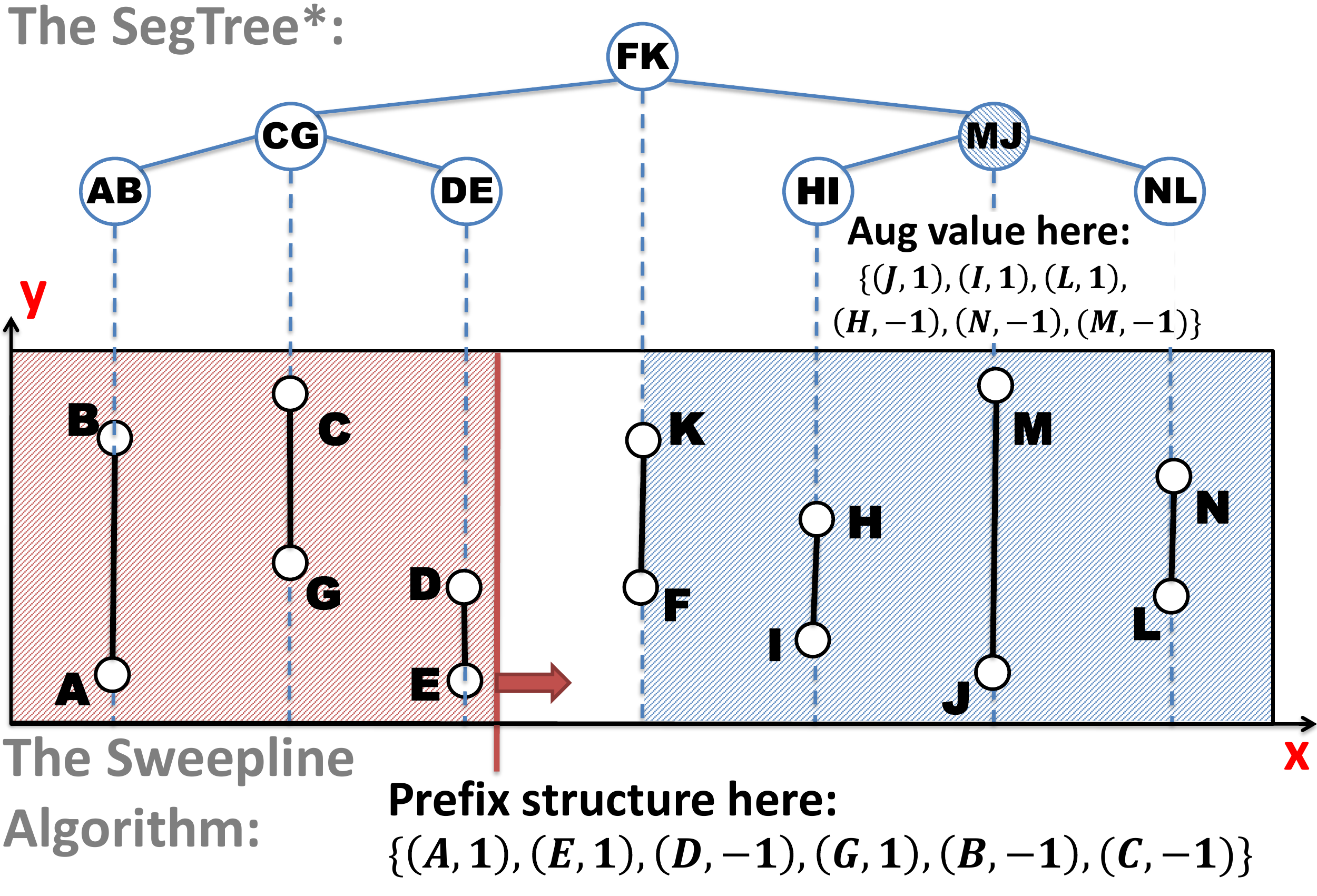}\\
  \end{tabular}
\caption{\textbf{An illustration of all data structures introduced in this paper} - The input data are shown in the middle rectangle. We show the tree structures on the top, and the sweepline algorithm on the bottom.
  All the inner trees (the augmented values or the \sweepstructure{}s) are shown as sets (or a pair of sets) with elements listed in sorted order. } \label{fig:expfigures}
  \label{fig:illustration}
\end{figure*}
\section{2D Segment Query}
\label{sec:segment}
Given a set of non-intersecting 2D segments, and a vertical segment $S_q$, segment query requires some information related to all the segments that cross $S_q$ to be reported.
We note a segment as its two endpoints $(l_i, r_i)$ where $x(l_i)\le x(r_i)$, and say it \emph{starts from $x(l_i)$} and \emph{ends at $x(r_i)$}.

In this paper we introduce a two-level map structure \emph{SegMap} addressing this problem (shown in Table \ref{tab:allstructures} and Figure \ref{fig:illustration} (b)).
The keys of the outer map are the x-coordinates of all endpoints of the input segments.
Each value stores the corresponding segment.
Each (sub-)outer map corresponds to an interval on the x-axis (from the leftmost to the rightmost key in the sub-map), noted as the \emph{x-range} of this map. The augmented value of an outer map is a pair of inner maps: $L(\cdot)$ (the \emph{\leftopen{}}) which stores all input segments starting outside of its x-range and ending inside (i.e., only the right endpoint is in its x-range), and symmetrically $R(\cdot)$ (the \emph{\rightopen{}}) with all segments starting inside but ending outside. We call them the \emph{\opensets{}} of the corresponding interval. The \opensets{} of an interval $u$ can be computed by combining the \opensets{} of its sub-intervals. In particular, suppose $u$ is composed of two contiguous intervals $u_l$ and $u_r$, then $u$'s \opensets{} can be computed by an associative function $f_{\text{seg}}$ as:
\begin{align}
\nonumber
\boldsymbol{f_{\text{seg}}}&\left(\langle L(u_l),R(u_l)\rangle,\langle L(u_r),R(u_r)\rangle\right) = \\
\label{eqn:segcombine}
&\langle L(u_l)\cup(L(u_r) \backslash R(u_l)), R(u_r)\cup(R(u_l)\backslash L(u_r)) \rangle
\end{align}
Intuitively, taking the \rightopen{} as an example, it stores all segments starting in $u_r$ and going beyond, or those stretching out from $u_l$ but \emph{not} ending in $u_r$. We use $f_{\text{seg}}$ as the combine function of the outer map.

The base function $g_{\text{seg}}$ of the outer map (as shown in Table \ref{tab:allstructures}) computes the augmented value of a single entry. For an entry $(x_k, (p_l, p_r))$, the interval it represents is the solid point at $x_k$. WLOG we assume $x_k=x(p_l)$ such that the key is the left endpoint. Then the only element in its \opensets{} is the segment itself in its \rightopen{}. If $x_k>x_v$ it is symmetric.

We organize all segments in an inner map sorted by their y-coordinates and augmented by the count, such that in queries, the range search on y-coordinates can be done in the corresponding inner maps.
We note that all segments in a certain inner tree must cross one common x-coordinate. For example, in the \leftopen{} of an interval $i$, all segments must cross the left border of $i$. Thus we can use the y-coordinates at this point to determine the ordering of all segments.
Note that input segments are non-intersecting, so this ordering of two segments is consistent at any x-coordinate.
The definition of such an inner map is in Table \ref{tab:allstructures} (the inner map $S_I$).
The construction of the two-level map \emph{SegMap} ($S_M$) from a list of segments $B=\{(p_l,p_r)\}$ can be done as follows:
\begin{align}
\nonumber
s_M=&S_M.\mb{Build}(B'), \text{ where }\\
\nonumber
B' &= \{(x(p_l),(p_l,p_r)), (x(p_r),(p_l,p_r)): (p_l,p_r)\in B\}
\end{align}

Assume the query segment is $(p_s, p_e)$, where $x(p_s)=x(p_e)=x_q$ and $y(p_s)<y(p_e)$.
The query will first find all segments that cross $x_q$, and then conduct a range query on $(y(p_s),y(p_e))$ on the y-coordinate among those segments. To find all segments that cross $x_q$, note that they are the segments starting before $x_q$ but ending after $x_q$, which are exactly those in the \rightopen{} of the interval $(-\infty,x_q)$. This can be computed by the function \mb{ALeft}.
The counting query can be done using the augmented map interface as:
\begin{align*}
\mb{SegQuery}&(s_M,p_s,p_e) = S_I.\mb{ARange} \\
&\left(\mb{ROpen}(S_M.\mb{ALeft}(s_M,x(p_t))), y(p_s),y(p_e)\right)
\end{align*}
where \mb{ROpen}$(\cdot)$ extracts the \rightopen{} from an open set pair. The list-all query can be answered similarly using $S_I$.\mb{Range} instead of $S_I$.\mb{ARange}.

We use augmented trees for inner maps. We discuss two implementations of the outer map: the augmented tree (which yields a segment-tree-like structure \emph{SegTree}) and the \sweepstructure{}s (which yields a sweepline algorithm \emph{SegSwp}). 
We also present another two-level augmented map (\emph{Segment* Map}) structure that can answer counting queries on axis-parallel segments in Appendix \ref{app:segcount}.


\subsection{The \segtree{}}
\label{sec:segmenttree}
If the outer map is implemented by the augmented tree, the \emph{SegMap} becomes very similar to a segment tree (noted as \emph{SegTree}). Previous work has studied a similar data structure~\cite{Chaselle1984,AtallahG86,aggarwal1988parallel}. We note that their version can deal with more types of queries and problems, but we know of no implementation work of a similar data structure.
Our goal is to show how to apply the simple augmented map framework to model the segment query problem, and show an efficient and concise parallel implementation of it.

In segment trees, each subtree represents an open interval, and the union of all intervals in the same level span the whole interval (see Figure \ref{fig:illustration} (b) as an example). The intervals are separated by the endpoints of the input segments, and the two children partition the interval of the parent. Our version is slightly different from the classic segment trees in that we also use internal nodes to represent a point on the axis. For example, a tree node $u$ denoting an interval $(l,r)$ have its left child representing $(l, k(u))$, right child for $(k(u),r)$, and the single node $u$ itself, is the solid point at it key $k(u)$.
For each tree node, the \emph{SegTree} tracks the \opensets{} of its subtree's interval, which is exactly the augmented value of the sub-map rooted at $u$. The augmented value (the \opensets{}) of a node can be generated by combining the \opensets{} of its two children (and the entry in itself) using Equation \ref{eqn:segcombine}.

\para{Answering Queries more efficiently. }
The \mb{ALeft} on the outer tree of \emph{SegTree} is costly, as it would require $O(\log n)$ \mb{Union} and \mb{Different} on the way. Here we present a more efficient query algorithm making use of the tree structure, which is a variant of the algorithm in \cite{Chaselle1984,AtallahG86}. Besides the \opensets{}, in each node we store two helper sets (called the \emph{\symdiff{}}): the segments starting in its left half and going beyond the whole interval ($R(u_l)\backslash L(u_r)$ as in Equation \ref{eqn:segcombine}), and vice versa ($L(u_r)\backslash R(u_l)$). These \symdiff{} are the side-effect of computing the \opensets{}. Hence in implementations we just keep them with no extra work.
Suppose $x_q$ is unique to all the input endpoints. The query algorithm first searches $x_q$ outer tree. Let $u$ be the current visited tree node. Then $x_q$ falls in either the left or the right side of $k(u)$. WLOG, assume $x_q$ goes right. Then all segments starting in the left half and going beyond the range of $u$ should be reported because they must cover $x_q$. All such segments are in $R(\lc{u})\backslash L(\rc{u})$, which is in $u$'s \symdiff{}. The range search on y-coordinates will be done on this \symdiff{} tree structure. After that, the algorithm goes down to $u$'s right child to continue the search recursively.
The cost of returning all related segments is $O(k+\log^2 n)$ for output size $k$, and the cost of returning the count is $O(\log^2 n)$.

\subsection{The Sweepline Algorithm}
\label{sec:segsweep}
If \sweepstructure{s} are used to represent the outer map, the algorithm becomes a sweepline algorithm \emph{SegSwp} (shown as $S_S$ in Table \ref{tab:allstructures}). We store at each endpoint $p$ the augmented value of the prefix of all points up to $p$. Because the corresponding interval is a prefix, the \leftopen{} is always empty. For simplicity we only keep the \rightopen{} as the \sweepstructure{}, which is all ``alive'' segments up to the current event point (a segment $(p_l,p_r)$ is alive at some point $x\in X$ iff $x(p_l)\le x \le x(p_r)$).

Sequentially, as the line sweeping through the plane, each left endpoint should trigger an insertion of its corresponding segment into the \sweepstructure{} while the right endpoints cause deletions. This is also what happens when a single point is plugged in as $u_r$ in Equation \ref{eqn:segcombine}. We use our parallel sweepline paradigm to parallelize this process.
In the batching step, we compute the augmented value of each block, which is the \opensets{} of the corresponding interval.
Basically, the \leftopen{} of an interval are segments with their right endpoints inside the interval, noted as $L$, and the \rightopen{} is those with left endpoints inside, noted as $R$.
In the sweeping step, the \sweepstructure{} is updated by the combine function $f_{\text{seg}}$, but only on the \rightopen{}, which is equivalent to first taking a \mb{Union} with $R$ and then a \mb{Difference} with $L$. Finally in the refining step, each left endpoint triggers an insertion and each right endpoint cause a deletion. This algorithm fits the sweepline abstraction in Theorem \ref{thm:sweeptheory}, so the corresponding bound holds.


\para{Answering Queries.} The \mb{ALeft} on the \sweepstructure{} is straightforward which is just a binary search of $x_q$ in the sorted list of x-coordinates.
In that \sweepstructure{} all segments are sorted by y-coordinates, and we search the query range of y-coordinates on that. 
In all, a query on reporting all intersecting segments costs $O(\log n +k)$ ($k$ is the output size), and a query on counting related segments costs $O(\log n)$.

%% file: app-rec.tex
\section{2D Rectangle Query}

Given a set of rectangles on the 2D planar, the rectangle query (also referred to as the \emph{orthogonal point enclosure query}) requires all rectangles containing a query point $p_q=(x_q, y_q)$ to be reported. Each rectangle $C=(p_l, p_r)$, where $p_l, p_r\in D$, is represented as its left-top and right-bottom vertices. We say the interval $[x(p_l), x(p_r)]$ and $[y(p_l), y(p_r)]$ are the \emph{x-interval} and \emph{y-interval} of $C$, respectively.

The rectangle query can be answered by a two-level map structure \emph{RecMap} ($G_M$ in Table \ref{tab:allstructures} and Figure \ref{fig:illustration} (c)), which is similar to the \emph{SegMap} as introduced in Section \ref{sec:segment}.
The outer map organizes all rectangles based on their x-intervals using a similar structure as the outer map of \emph{SegMap}.
The keys of the outer map are the x-coordinates of all endpoints of the input rectangles, and the values are the rectangles.
The augmented value of a (sub-)outer map is also the \emph{\opensets{}} as defined in \emph{SegMap}, which store the rectangles that partially overlap the x-range of this sub-map. The combine function is accordingly the same as the segment map.

Each inner map of the \emph{RecMap} organizes the rectangles based on their y-intervals.
All the y-intervals in an inner tree are organized in an \emph{interval tree} (the term \emph{interval tree} refers to different definitions in the literature. We use the definition in \cite{CLRS}).
The interval tree is an augmented tree (map) structure storing a set of 1D intervals sorted by the left endpoints, and augmented by the maximum right endpoint in the map. Previous work \cite{pam} has studied implementing interval trees using the augmented map interface. It can report all intervals crossing a point in time $O(\log n +k\log(n/k+1))$ for input size $n$ and output size $k$.

\emph{RecMap} answers the enclosure query of point $(x_q, y_q)$ using a similar algorithm as \emph{SegMap}. The query algorithm first finds all rectangles crossing $x_q$ by computing the \rightopen{} $R$ in the outer map up to $x_q$ using \mb{ALeft}, which is an interval tree.
The algorithm then select all rectangles in $R$ crossing $y_q$ by applying a list-all query on the interval tree.

Using interval trees as inner maps does not provide efficient interface for counting queries. We use the same inner map as in \emph{SegMap*} for counting queries. The corresponding map structure (\emph{RecMap*}) is presented in Appendix \ref{app:reccount}.

We use augmented trees for inner maps (the interval trees). We discuss two implementations of the outer map: the augmented tree (which yields a multi-level tree structure) and the \sweepstructure{}s (which yields a sweepline algorithm).

\subsection{The Multi-level Tree Structure}
\label{sec:rectangletree}
If the outer map is implemented by the augmented tree, \emph{RecMap} becomes a multi-level tree structure. The outer tree structure is similar to the \segtree{}, and we use the same trick of storing the \symdiff{} in the tree nodes to accelerate queries.
The cost of a list-all query is $O(k\log (n/k+1)+\log^2 n)$ for output size $k$.

\hide{
In this tree each node, or subtree, represents an interval, and we track all the segments that have one endpoint in this range. More specifically, for each range $s_i = [l_i, r_i]$, we keep two sets: $L(s_i)$ which stores all segments starting from outside on its left and ending in the interval (i.e., only the right endpoint in the interval), and $R(s_i)$ which stores all segments starting inside the range and ending outside on its right.

In this tree all nodes $v_i$ represent disjoint intervals $[l(v_i), r(v_i)]$, which is a partition of the whole range on the real axis. All the nodes are sorted by the key $k(v_i)$, which can be any value inside the range (i.e., as long as $l(v_i)\le k_i \le r(v_i)$). In addition, the augmented value of each node, which should represent some property in its subtree, is the interval from the leftmost key in its subtree to the rightmost one --- i.e., the union of all the disjoint intervals in its subtree.
}


\subsection{The Sweepline Algorithm}
\label{sec:recsweep}
If we use \sweepstructure{s} to represent the outer map, the algorithm becomes a sweepline algorithm ($G_S$ in Table \ref{tab:allstructures}). The skeleton of the sweepline algorithm is the same as \emph{SegSwp}--- the \sweepstructure{} at event point $x$ stores all ``alive'' rectangle at $x$. The combine function, fold function and update function are defined similar as in \emph{SegSwp}, but onto inner maps as interval trees.
This algorithm also fits the sweepline abstraction in Theorem \ref{thm:sweeptheory}, so the corresponding bound holds.

To answer the list-all query of point $(x_q,y_q)$, the algorithm first finds the \sweepstructure{} $t_q$ at $x_q$, and applies a list-all query on the interval tree $t_q$ at point $y_q$. The cost is $O(\log n+k\log(n/k+1))$ per query for output size $k$.


%% file: exp.tex
\section{Experiments}
\label{sec:exp}
We implement our algorithms for range, segment and rectangle queries using a parallel augmented map library (PAM)~\cite{pam}, which supports efficient functions for augmented maps using augmented trees.   We also implement the abstract parallel sweepline paradigm as described in Section \ref{sec:augsweep}.   Using PAM each of our data structures only need about 100 lines of code. Some examples are given in Appendix \ref{app:codeexample}. We plan to release our code on GitHub. More details about PAM and the implementation of the sweepline paradigm are given in Appendix~\ref{app:pam} and~\ref{app:sweepcode}. The PAM library supports persistence by path-copying. We run experiments on a 72-core Dell R930 with 4 x Intel(R)
Xeon(R) E7-8867 v4 (18 cores, 2.4GHz and 45MB L3 cache) with 1TB memory.  Each core is
2-way hyperthreaded giving 144 hyperthreads.  Our code was compiled
using the g++ 5.4.1 compiler which supports the Cilk Plus extensions.
We compile with \ctext{-O2}.

For range queries, we test \emph{RangeTree} in Section \ref{sec:rangetree} and \emph{RangeSwp} in Section \ref{sec:rangesweep}.
For segment queries, we test \emph{SegTree} in Section \ref{sec:segmenttree}, \emph{SegSweep} in Section \ref{sec:segsweep}, the counting versions \emph{SegTree*} and \emph{SegSwp*} in Appendix \ref{app:segcount}.
For rectangle queries, we test \emph{RecTree} in Section \ref{sec:rectangletree}, the \emph{RecSwp} in Section \ref{sec:recsweep} and the counting versions \emph{RecTree*} and \emph{RecSwp*} in Section \ref{app:reccount}. We use integer coordinates.
We test queries on both counting and list-all queries.   On list-all queries, since the cost is largely affected by the output size, we test a small and a large query window with average output size of less than $10$ and about $10^6$ respectively. We accumulate query time over $10^3$ large-window queries and over $10^6$ small window queries. For counting queries we accumulate $10^6$ queries picked uniformly at random. For parallel queries we process all queries in parallel using a parallel for-loop. The sequential algorithms tested in this paper are directly running the parallel algorithms on one core.
We use $n$ for the input size, $k$ the output size, $p$ the number of threads.  For the sweepline algorithms we set $b=p$, and do not apply the sweepline paradigm recursively to blocks.

We compare our sequential versions with two C++ libraries CGAL \cite{CGAL} and Boost \cite{boost}. CGAL provides the same range tree \cite{cgal-range} data structure as ours, and the segment tree \cite{cgal-range} in CGAL implements the 2D rectangle query.
Boost provides an implementation of R-trees, which can be used to answer range, segment and rectangle queries. CGAL and Boost only supports list-all queries. For Boost, we also parallelize the queries using OpenMP. CGAL uses some shared state in queries thus the queries cannot be parallelized trivially. We did not find comparable parallel implementations in C++, so we compare our parallel query performance with Boost, and also compare the parallel performance of multi-level tree structures and sweepline algorithms with each other.

In the rest of this section we show results for range queries and segment queries and comparisons across all tested structures.
We will show that our implementations are highly-parallelized. On 72 cores with 144 hyperthreads, our implementations achieve 32-70x speedup in construction, and 60-130x speedup in queries. The overall sequential performance (construction and query) of our implementations is comparable or outperforms existing implementations.
In Appendix~\ref{app:exp} we present more results and discussions on experimental results of the four data structures for counting queries, our dynamic updates on range trees, the scalability curves, and space (memory) consumption of all structures.

\subsection{2D Range Queries}
We test our implementation on \emph{RangeTree} and \emph{RangeSwp}, for both counting and list-all queries, sequentially and in parallel. We show the running time on $n=10^8$ in Table \ref{tab:rangetable}.
We also give the the scalability curve on $n=10^8$ points in Appendix \ref{app:exp}, Figure \ref{fig:expfigures} (a).
\hide{We compare our sequential version with two C++ libraries: the range tree in CGAL \cite{CGAL} and the R-tree in Boost \cite{boost}, which are both sequential and only supports list-all. We cannot find a comparable parallel implementation in C++, so for parallel performance, we just compare the two versions (\emph{RangeSwp} vs. \emph{RangeTree}) of our implementations.}

\para{Sequential Construction.} \emph{RangeTree} and \emph{RangeSwp} have similar performance and significantly outperform CGAL (2x faster), and is slightly faster than Boost (1.3-1.5x faster). Among all, \emph{RangeTree} is the fastest in construction. We guess the reason of the faster construction of our \emph{RangeTree} than CGAL is that their implementation makes two copies on data (once in merging and once to create tree nodes) while ours only copies the data once.

\para{Parallel Construction.} \emph{RangeTree} achieves a 63-fold speedup on $n=10^8$ and $p=144$. \emph{RangeSwp} has relatively worse parallel performance, which is a 33-fold speedup, and 2.3x slower than the \emph{RangeTree}. This is likely because of its worse theoretical depth ($\tilde{O}(\sqrt{n})$ vs. $O(\log^2 n)$). Another reason is that more threads means more blocks, introducing more overhead in batching and folding. As for \emph{RangeTree} not only the construction is highly-parallelized, but the combine function (\mb{Union}) is also parallel. Figure \ref{fig:expfigures}(a) shows that both \emph{RangeTree} and \emph{RangeSwp} scale up to 144 threads.

\para{Query Performance.} In counting queries, \emph{RangeSwp} shows the best performance in both theory and practice. On list-all queries, \emph{RangeSwp} is much slower than the other two range trees when the query window is large, but shows better performance for small windows. These results are consistent with their theoretical bounds. Boost R-tree is 1.5-26x slower than the our implementations, likely because of lack of worst-case theoretical guarantee in queries. The speedup numbers for queries are above 52 because they are embarrassingly parallel, and speedup numbers of our implementations are slightly higher than Boost.


\begin{table*}
\small
  \centering
    \begin{tabular}{c@{}c@{ }||rrr|rrr|rrr|rrr}
    \hline
       \multirow{2}{*}{$\boldsymbol{n}$}&\multirow{2}{*}{\textbf{Algorithm}}   &  \multicolumn{3}{c|}{\textbf{Build, s}} &        \multicolumn{3}{c|}{\textbf{Counting Query, $\boldsymbol{\mu}$s}} &      \multicolumn{3}{c|}{\textbf{List-all (small), $\boldsymbol{\mu}$s}} &      \multicolumn{3}{c}{\textbf{List-all (large), ms}}\\
\cline{3-14}
          && \multicolumn{1}{c}{\textbf{Seq.}} & \multicolumn{1}{c}{\textbf{Par.}} & \multicolumn{1}{c|}{\textbf{Spd.}} &  \textbf{Seq.} &\textbf{Par.} & \textbf{Spd.} &  \multicolumn{1}{c}{\textbf{Seq.}} & \multicolumn{1}{c}{\textbf{Par.}} & \multicolumn{1}{c|}{\textbf{Spd.}}&  \multicolumn{1}{c}{\textbf{Seq.}}&\textbf{Par.} & \textbf{Spd.}\\
    \hline\hline
\multirow{4}{*}{$\boldsymbol{10^8}$}&{\bf RangeSwp} &   243.89  & 7.30      & 33.4  &    12.74 &    0.15   &  86.7 &  11.44  &   0.13      &  85.4    &   213.27   & 1.97 & 108.4  \\
&{\bf RangeTree} &     200.59 &       3.16 & 63.5   &      61.01 &      0.75 &  81.1  &     17.07 &      0.21 & 80.5  &     44.72 &      0.69 & 65.2   \\
&{\bf Boost} &    315.34 &          - &          - &          - &          - &          - &    25.41 &          0.51 &          49.8&    1174.40 &          22.42 &          52.4 \\
&{\bf CGAL} &    525.94 &          - &          - &          - &          - &          - &    153.54 &          - &          - &    110.94 &          - &          - \\
    \hline
\multirow{3}{*}{$\boldsymbol{5\times 10^7}$}&{\bf SegSwp} &  254.49   &   7.20    &   35.3 &    6.78  &      0.09 & 75.3   &  6.18     &   0.08    &   77.2    &   255.72    & 2.65    &   96.5 \\
&{\bf SegTree} &  440.33   &  6.79   &   64.8 &  50.31   &  0.70  & 71.9    &  39.02  & 0.48   & 81.7   &  123.26    &  1.99    &   61.9    \\
&{\bf Boost} &  179.44   &  -   &   - &  -   &  -  & -    &  7421.30  & 113.09   & 65.6   &  998.20    &  23.21    &   43.0    \\
\hline
\multirow{2}{*}{$\boldsymbol{5\times 10^7}$}&{\bf SegSwp*} &  233.19   &  7.16   &  32.6  &   7.44  & 0.11  &67.6   &          - &          - &          - &          - &          - &          - \\
&{\bf SegTree*} &    202.01 &      3.21 &  63.0  &       33.58 &      0.40 & 83.8   &          - &          - &          - &          - &          - &          - \\
\hline
\multirow{3}{*}{$\boldsymbol{5\times 10^7}$}&{\bf RecSwp} &   241.51  & 6.76      &  35.7  &       - &     - & -   &   8.34   &  0.10           &          83.4 &          575.46 &          5.91 &          97.4 \\
&{\bf RecTree} &    390.98 &   6.23    & 62.8   &       - &     - &  - &43.57   &       0.58 &   75.1 &         382.26 &     5.35 &          71.4\\
&{\bf Boost} &    183.65 &      - & -   &       - &     - &  - & 52.22  &  0.94       &  55.6 &          706.50 &          11.10 &          63.6\\
$\boldsymbol{5\times 10^6}$&{\bf CGAL$^{[1]}$} &    398.44 &      - & -   &       - &     - &  - & 90.02  &  -       &   - &         4412.67 &          - &          -\\
\hline
\multirow{2}{*}{$\boldsymbol{5\times 10^7}$}&{\bf RecSwp*} & 585.18    &  12.37     & 47.32   &      6.56 &  0.05 & 126.1   &  -    &  -      &          - &          - &       - &          - \\
&{\bf RecTree*} & 778.28    &  11.34     & 68.63   &       39.75 &    0.35 &  113.6 &-   &       - &  - &          - &          - &          -\\
\hline
    \end{tabular}\vspace{-.1in}
    \caption{\textbf{The running time of all data structures} - ``Seq.'', ``Par.'' and ``Spd.'' refer to the sequential, parallel running time and the speedup. [1]: Result of CGAL is shown as on input size $5\times 10^6$. On $5\times 10^7$, CGAL did not finish in one hour.}
  \label{tab:rangetable}%
  \vspace{-.1in}
\end{table*}%
\hide{For query-all we also adjust window size of the query (see details in Section \ref{sec:expsetting}). Seq. and Par. represent sequential and parallel time respectively. We give the CGAL range tree (sequential) time for comparison.}

\subsection{2D Segment Query}
We show the running times on segment queries using \emph{SegSwp}, \emph{SegTree}, \emph{SegSwp*} and \emph{SegTree*},
on $n=5\times 10^7$ in Table \ref{tab:rangetable}.   We also show the parallel speedup in Figure \ref{fig:expfigures}(b). Note that for these structures on input size $n$ (number of segments), $2n$ points are processed in map. Thus we set input size to be $n=5\times 10^7$ for comparison with the maps for range queries.
\hide{We compare our sequential version with Boost R-tree. We did not find comparable parallel C++ implementations on the same problem so we just compare our implementations with each other.} We discuss the performance of \emph{SegTree*} and \emph{SegSwp*} in Appendix \ref{app:exp}.

\para{Sequential Construction.} Sequentially, Boost is 1.4x faster than \emph{SegSwp} and 2.4x than \emph{SegTree}. This is likely because of R-tree's simpler structure. However, Boost is 4-200x slower in sequential queries than our implementations.
\emph{SegTree} is the slowest in construction because it stores four sets (the \opensets{} and the \symdiff{}) in each node, and calls two \mb{Difference} and two \mb{Union} functions in each combine function.


\para{Parallel Construction.} For $n=5\times 10^7$ input segments, \emph{SegTree} is slightly faster than \emph{SegSwp} in parallel construction. Considering that \emph{SegTree} is about 1.7x slower than \emph{SegSwp} sequentially, the good parallel performance comes from its good scalability (64x speedup). The lack of parallelism of \emph{SegSwp} is of the same reason as \emph{RangeSwp}.

\para{Query Performance.} In the counting query and list-all query on small window size, \emph{SegSweep} is significantly faster than \emph{SegTree} as would be expected from its better theoretical bound.
As for list-all on large window size, although \emph{SegTree} and \emph{SegSwp} have similar theoretical cost (output size $k$ dominates the cost), \emph{SegTree} is faster than \emph{SegSwp} both sequentially and in parallel.
This might have to do with locality.  In the sweepline algorithms, the tree nodes even in one \sweepstructure{} were created at different times because of path-copying, and thus are not contiguous in memory, leading to bad cache performance.
Both \emph{SegSwp} and \emph{SegTree} have better query performance than Boost R-tree (8.7-1200x faster in parallel). Also, R-tree does not take the advantage of smaller query window.
Comparing the sequential query performance on large windows with on small windows, on outputting about $10^6x$ less points, \emph{SegTree} and \emph{SegSwp} are $3000$x and $40000$x faster respectively, while Boost R-tree is only 130x faster. Our implementations on small windows is not $10^6$x as faster as on large windows because on small windows, the $\log n$ or $\log^2 n$ term dominates the cost. This illustrates that the bad query performance of R-tree comes from lack of worst-case theoretical guarantee.
The query speedup of our implementations is over 61.


\subsection{2D Rectangle Query}
We show the running times on rectangle queries using \emph{RecSwp}, \emph{RecTree}, \emph{RecSwp*} and \emph{RecTree*},
on $n=5\times 10^7$ in Table \ref{tab:rangetable}.   We also show the parallel speedup in Figure \ref{fig:expfigures}(c). We discuss the performance of \emph{RecTree*} and \emph{RecSwp*} in Appendix \ref{app:exp}.

\para{Sequential Construction.} Sequentially, the three implementations have close performance as in the segment queries, in which Boost is faster in construction than the other two (1.6-2.1x), but is much slower in queries, and
\emph{RecTree} is slow in construction because of its more complicated structure. CGAL did not finish construction on $n=5\times 10^7$ in one hour, and thus we present the results on $n=5\times 10^6$. In this case, CGAL has a performance slightly worse than our implementations even though our input size is 10x larger.

\para{Parallel Construction.} The parallel performance is similar to the segment queries, in which \emph{RecTree} is slightly faster than \emph{RecSwp} because of good scalability (62x speedup).

\para{Query Performance.} In list-all queries on a small window size, \emph{RecSwp} is significantly faster than other implementations due to its better theoretical bound. Boost is 1.2-9x slower than our implementations when query windows are small, and is 1.2-2x slower when query windows are large, both sequentially and in parallel.
The query speedup of our implementations is over 71.

\subsection{Other experiments} We also give the results and discussions on the performance of the data structures for counting queries in Appendix \ref{app:countingexp}, the scalability in Appendix \ref{app:scale}, our update function for range trees in Appendix \ref{app:update} and space (memory) consumption in Appendix \ref{app:space}.

\subsection{Summary}
All results in Table \ref{tab:rangetable} are on $10^8$ elements, so we briefly compare all of them.
The sweepline algorithms usually have better construction time sequentially, but in parallel are slower than two-level tree structures.
This has to do with the better scalability of the two-level trees.
For the construction of two-level trees, with properly defined augmentation, the construction is automatically done by the augmented map constructor in PAM, which is highly-parallelized (polylog depth). For the sweepline algorithms, the parallelism comes from the blocking-and-batching process, with a $\tilde{O}(\sqrt{n})$ depth.
Most of the implementations have similar construction time. Sequentially \emph{SegTree} and \emph{RecTree} is much slower than the others, because they store more information in each node (four maps) and has a more complicated combine function. The speedup numbers of all sweepline algorithms are close at about 30-35x, and all two-level trees at about 62-68x.

Generally the sweepline algorithms are better in counting queries and small window queries (when the output size $k$ does not dominate the cost) because of better theoretical bound. On large window queries, the two-level tree structures generally are faster since theoretically the output size $k$ dominates the cost, and the two-level trees have better locality.

Comparing to CGAL and Boost, our implementations outperforms CGAL in both construction and queries. Overall, Boost R-tree can be about 2x faster than our algorithms in construction, but is always slower (1.6-1400x) in queries, likely because of its lack of worst-case theoretical guarantee.

%% file: conclusion.tex
\section{Conclusion}
\label{sec:conclusion}
This paper we develop implementations of a broad set of parallel algorithms for range, segment and rectangle queries that are very much faster and simpler than the previous implementations as well as theoretically efficient.
We did this by using a the augmented map framework. Based on different representations (augmented trees and \sweepstructure{}s), we design both multi-level trees and sweepline algorithms addressing range, segment and rectangle queries. We implement all algorithms in these paper and test the performance both sequentially and in parallel. Experiments show that our algorithms achieve good speedup, and have good performance even running sequentially. The overall performance considering both construction and queries of our implementations outperforms existing sequential libraries such as CGAL and Boost.

\hide{
In this paper we study data structures for range queries and segment queries especially in the parallel setting. We propose a simple framework for these problems using the abstraction of augmented maps, as well as a parallel sweepline paradigm. By adapting the problems into such frameworks, we can design simple, efficient and parallel algorithms for these problems. We implement all algorithms in these paper and test the performance both sequentially and in parallel. Experiments show that our algorithm achieve good speedup, and are also good even running sequentially, and outperforms existing sequential implementations such as CGAL. } 

%% file: appendix.tex
\begin{table*}[!t]
{\small
\begin{tabular}{r@{}r@{ }@{}c@{ }@{ }c@{ }@{ }c@{ }@{ }l@{ }@{}l@{ }@{ }l@{ }@{}l@{ }@{ }l@{ }@{ }l@{ }@{ }l@{ }@{ }l@{ }}
\hline
\multicolumn{13}{@{}l}{\textbf{\normalsize{* Segment Count Query:}}}\\
  \multicolumn{2}{@{}r@{}}{\bf(Inner Map)$\boldsymbol{S^*_I}$} &=& \textbf{AM}&(& $\boldsymbol{K}$: $Y$; &$\boldsymbol{\prec}$: $<_{Y}$;  &$\boldsymbol{V}$: $D\times D$; &$\boldsymbol{A}$: $\mathbb{Z}$; &$\boldsymbol{g}$: $(k,v)\mapsto 1$; &$\boldsymbol{f}$: $+_{\mathbb{Z}}$; &$\boldsymbol{I}$: $0$ &)\\
{\bf -}  &\bf SegMap* $\boldsymbol{S^*_M}$&=& \textbf{AM}&(&  $\boldsymbol{K}$: $X$; & $\boldsymbol{\prec}$: $<_{X}$; &$\boldsymbol{V}$: $Y\times Y$;&$\boldsymbol{A}$: $S^*_I$;&$\boldsymbol{g}$: $g^*_{\text{seg}}$&$\boldsymbol{f}$: $S^*_I$.union& $\boldsymbol{I}$: $\emptyset$ &\multicolumn{1}{@{}l@{}}{)}\\
\multicolumn{1}{l}{} & &\multicolumn{11}{@{}l}{$\boldsymbol{g^*_{\text{seg}}(x,(l,r))}=C_I$.build$\left(\{(l,1), (r,-1)\}\right)$} \\
 {\bf -}  &\bf SegSwp* $\boldsymbol{S^*_S}$&=& \textbf{PS}&(& $\boldsymbol{P}$: $D\times D$;&$\boldsymbol{\prec}$: $<_Y$;&$\boldsymbol{T}$: $C_I$;& $\boldsymbol{t_0}$: $\emptyset$;&$\boldsymbol{h}$: $h^*_{\text{seg}}$;& $\boldsymbol{\funca}$: $\rho^*_{\text{seg}}$;& $\boldsymbol{f}$: $C_I$.union &\multicolumn{1}{@{}l@{}}{)}\\
   \multicolumn{1}{l}{}   &\multicolumn{12}{@{}r}{$\boldsymbol{h^*_{\text{seg}}(t,(p_l,p_r))}=C_I$.union$(t,\{(y(p_l),1),(y(p_r),-1)\})$, \quad $\boldsymbol{\funca^*_{\text{seg}}}(\{(p_l,p_r)\})=C_I.$build$(\{(y(p_l),1),(y(p_r),-1)\})$ } \\
\hline
\multicolumn{13}{@{}l}{\textbf{\normalsize{* Rectangle Count Query:}}}\\
  \multicolumn{2}{@{}r@{}}{\bf(Inner Map)$\boldsymbol{G^*_I}$} &=& \multicolumn{10}{l}{$S^*_I$}\\
  \bf - &{\bf(RecMap*)$\boldsymbol{G^*_M}$} &=& \multicolumn{10}{l}{similar as $G_M$, but use $G^*_I$ as inner maps}\\
  \bf -&{\bf(RecSwp*)$\boldsymbol{G^*_S}$} &=& \multicolumn{10}{l}{similar as $G_S$, but use $G^*_I$ as \sweepstructure{s}}\\

\hline
\end{tabular}
}
\caption{\textbf{Definitions of \emph{SegMap*} and \emph{RecMap*}} - $X$ and $Y$ are types of x- and y-coordinates. $D=X\times Y$ is the type of a point.}
\label{tab:appallstructures}\vspace{-.2in}
\end{table*}

\section{Data Structures for Segment Counting Queries}
\label{app:segcount}
In this section, we present a simple two-level augmented map \emph{SegMap*} structure to answer segment count queries (see the \emph{Segment Count Query} in Table \ref{tab:appallstructures} and Figure \ref{fig:illustration} (d)). This map structure can only deal with axis-parallel input segments.
For each input segment $(p_l,p_r)$, we suppose $x(p_l)=x(p_r)$, and $y(p_l)<y(p_r)$.
We organize the x-coordinates in the outer map, and deal with y-coordinates in the inner trees. We first look at the inner map. For a set of 1D segments, a feasible solution to count the segments across some query point $x_q$ is to organize all end points in sorted order, and assign signed flags to them as values: left endpoints with $1$, and right endpoints with $-1$. Then the prefix sum of the values up to $x_q$ is the number of alive segments.
To efficiently query the prefix sum we can organize all endpoints as keys in an augmented map, with values being the signed flags, and augmented values adding values. We call this map the \countmap{} of the segments.

To extend it to 2D scenario, we use a similar outer map as range query problem. On this level, the x-coordinates are keys, the segments are values, and the augmented value is the \countmap{} on y-coordinates of all segments in the outer map.
The combine function is \mb{Union} on the \countmap{}s. However, different from range maps, here each tree node represents \emph{two} endpoints of that segment, with signed flags $1$ (left) and $-1$ (right) respectively, leading to a different base function ($g^*_{\text{seg}}$).

We maintain the inner maps using augmented trees. By using augmented trees and \sweepstructure{s} as outer maps, we can define a two-level tree structure and a sweepline algorithm for this problem respectively.
Each counting query on the \countmap{} of size $m$ can be done in time $O(\log m)$. In all, the rectangle counting query cost time $O(\log^2 n)$ using the two-level tree structure \emph{SegTree*}, and $O(\log n)$ time using the sweepline algorithm \emph{SegSwp*}.

We present corresponding definition and illustration on both the multi-level tree structure and the sweepline algorithm in Table \ref{tab:allstructures} and Figure \ref{fig:illustration} (d).

\section{Data Structures for Rectangle Counting Queries}
\label{app:reccount}
In this section, we extend the \emph{RecMap} structure to \emph{RecMap*} for supporting fast counting queries. We use the exactly outer map as \emph{RecMap}, but use base and combine functions on the corresponding inner maps. The inner map, however, is the same map as the \emph{\countmap} in \emph{SegMap*} ($S^*_I$ in Table \ref{tab:appallstructures}).
Then in queries, the algorithm will find all related inner maps, which are \countmap{s} storing all y-intervals of related rectangles. To compute the count of all the y-intervals crossing the query point $y_q$, the query algorithm simply apply an \mb{ALeft} on the \countmap{s}.

We maintain the inner maps using augmented trees. By using augmented trees and \sweepstructure{s} as outer maps, we can define a two-level tree structure and a sweepline algorithm for this problem respectively.
In all, the rectangle counting query cost time $O(\log^2 n)$ using the two-level tree structure \emph{RecTree*}, and $O(\log n)$ time using the sweepline algorithm \emph{RecSwp*}.

We present corresponding definition and illustration on both the multi-level tree structure and the sweepline algorithm in Table \ref{tab:allstructures}. The outer representation of \emph{RecMap*} is of the same format as \emph{RecMap} as shown in Figure \ref{fig:illustration} (c).

\section{Extend RangeSwp to Report All Points}
\label{app:rangesweeplist}
In the main body we have mentioned that by using a different augmentation, we can adjust the sweepline algorithm for range query to report all queried points.
It is similar to \emph{RangeSwp}, but instead of the count, the augmented value is the maximum x-coordinate among all points in the map. To answer queries we first narrow the range to the points in the inner map $t_R$ by just searching $x_R$. In this case, $t_R$ is an augmented tree structure.
Then all queried points are those in $t_R$ with x-coordinates larger than $x_L$ and y-coordinate in $[y_L,y_R]$. We still conduct a standard range query in $[y_L, y_R]$ on $t_R$, but adapt an optimization that if the augmented value of a subtree node is less than $x_L$, the whole subtree is discarded.
Otherwise, at least part of the points in the tree would be relevant and we recursively deal with its two subtrees.
\hide{
\begin{algorithm}[!th]
\caption{\func{ReportAll} ($t$,$x_L$, $y_L$, $y_R$)}
\label{algo:rangelistquery}
\KwIn{Current tree node $t$, the lower bound of the x-coordinate of the query window $x_L$, the y-range $[y_L,y_R]$}
\KwOut{Return all points in the subtree rooted at $r$ that have x-coordinate no less than $x_L$ and y-coordinate in range $[y_L,y_R]$}
\vspace{0.5em}
    \While{$r$ is not empty}{
      \If{$y(t)<y_L$} {$t\gets \rc{t}$; continue}
      \If{$y(t)>y_L$} {$t\gets \lc{t}$; continue}
    }
    $A\gets\emptyset$\\
	\lIf {($x(t)\ge x_L$)} {$A\gets A+e(t)$}
	$p \gets \lc{t}$\\
	\While {$p$} {
		\lIf {$y(p)<y_L$} {$p \gets \rc{p}$; continue}
		\lIf {$x(p)>x_L$} {$A\gets A+e(p)$}
		\func{FilterX}($\rc{p}$, $x_L$, $A$);
		$p\gets\lc{p}$
	}
	$p \gets \rc{t}$\\
	\While {$p$} {
		\lIf {$y(p)>y_R$} {$p \gets \lc{p}$; continue}
		\lIf {$x(p)>x_L$} {$A \gets A+e(p)$}
		\func{FilterX}($\lc{p}$, $x_L$, $A$);
		$p\gets\lc{p}$
	}
\Return {$A$}\\
\medskip
\SetKwProg{myfunc}{Function}{}{}
    \myfunc{\bf \func{FilterX}($t$, $x$, $A$)} {
      \lIf {$\mathcal{A}(t)<x$} {\Return}
      \lIf {$x(r)>x$}{$A\gets A+e(t)$}\label{line:visitmore}
      \func{FilterX}($\lc{t}$, $x$, $A$); \\
      \func{FilterX}($\rc{t}$, $x$, $A$);
    }
\end{algorithm}}

Now we analyze the cost of the query algorithm. Assume that the output size is $k$.
The total cost is asymptotically the number of tree nodes the algorithm visits, which is asymptotically the number of all reported points and their ancestors. For $k$ nodes in a balanced tree of size $n$, the number of all its ancestors is equivalent to all the nodes visited by the \mb{Union} function based on split-join model~\cite{join} when merging this tree with a set of the $k$ nodes. When using AVL trees, red-black trees, weight-balanced trees or treaps, the cost of the \mb{Union} function is $O(k\log (n/k+1))$. Detailed proof for the complexity of the \mb{Union} function can be found in~\cite{join}.

\section{Proof for Corollary \ref{coro:depth}}
\label{app:rangesweepdepth}
\begin{proof}
To reduce the depth of the parallel sweepline paradigm mentioned in Section \ref{sec:augsweep}, we adopt the same algorithm as introduced in Theorem \ref{thm:sweeptheory}, but in the last refining step, repeatedly apply the same algorithm on each block.
If we repeat for a $c$ of rounds, for the $i$-th round, the work would be the same as splitting the total list into $k^i$ blocks. Hence the work is still $O(n\log n)$ every round. After $c$ rounds the total work is $O(cn\log n)$.

For depth, notice that the first step costs logarithmic depth, and the second step, after $c$ iterations, in total, requires depth $\tilde{O}(cb)$ depth. The final refining step, as the size of each block is getting smaller and smaller, the cost of each block is at most $O(\frac{n}{b^i}\log n)$ in the $i$-th iteration. In total, the depth is $\tilde{O}\left(cb+\frac{n}{b^c}\right)$, which, when $b=c^{\frac{c}{c+1}}n^{\frac{1}{c+1}}$, is $\tilde{O}(n^{1/(c+1)})$. Let $\epsilon = 1/(c+1)$, which can be arbitrary small by adjusting the value of $c$, we can get the bound in Corollary \ref{coro:depth}.
\end{proof}

Specially, when $c=\log n$, the depth will be reduced to polylogarithmic, and the total work is accordingly $O(n\log^2 n)$. This is equivalent to applying a recursive algorithm (similar to the divide-and-conquer algorithm of the prefix-sum problem). Although the depth can be poly-logarithmic, it is not work-efficient any more. If we set $c$ to some given constant, the work and depth of this algorithm are $O(n\log n)$ and $O(n^\epsilon)$ respectively.

\section{Dynamic Update on Range Trees Using Augmented Map Interface}
\label{app:wbfast}
The tree-based augmented map
interface supports insertions and deletions (implementing the
appropriate rotations).  This can be used to insert and delete on the augmented tree interface.
However, by default this requires updating the augmented
values from the leaf to the root, for a total of $O(n)$ work.
Generally, if augmented trees are used to support augmented maps, the insertion function will re-compute the augmented values of all the nodes on the insertion path, because inserting an entry in the middle of a map could completely change the augmented value.
In the range tree, the cost is $O(n)$ per update because the combine function (\mb{Union}) has about linear cost.
To
avoid this we implemented a version of ``lazy'' insertion/deletion that applies when the
combine function is commutative.  Instead of recomputing the augmented
values it simply adds itself to (or removes itself from) the augmented values along the
path using $f$ and $g$. This is similar to the standard range tree update algorithm~\cite{lueker1978data}.

The amortized cost per update is $O(\log^2 n)$ if the tree is weight-balanced.
Here we take the insertion as an example, but similar methodology can be applied to any mix of insertion and deletion sequences (to make deletions work, one may need to define the inverse function $f^{-1}$ of the combine function).
Intuitively, for any subtree of size $m$, imbalance occur at least every $\Theta(m)$ updates, each cost $O(m)$ to rebalance. Hence the amortized cost of rotations per level is $O(1)$, and thus the for a single update, it is $O(\log n)$ (sum across all levels). Directly inserting the entry into all inner trees on the insertion path causes $O(\log n)$ insertions to inner trees, each cost $O(\log n)$. In all the amortized cost is $O(\log^2 n)$ per update.
\hide{
This can be directly supported by the augmented map interface (and PAM). When the augmented map is supported by the augmented tree, and when the combine function is commutative, a ``lazy'' insertion function can be applied. First, at each node on the insertion path, we call $f$ on the original augmented value and the new element to be inserted. The correctness is guaranteed by the commutativity of $f$. When imbalance occurs, we re-compute the augmented values of the related nodes. Since the \mb{Union} function is accumulative, we can use this lazy insert function and get the same asymptotical bound as the standard range tree update algorithm.}

Similar idea of updating multi-level trees in (amortized) poly-logarithmic time can be applied to \emph{SegTree*}, \emph{RecTree} and \emph{RecTree*}. For \emph{SegTree}, the combine function is not communicative, and thus update may be more involved than simply using the interface of lazy-insert function.

\section{More Experimental results}
\label{app:exp}
\begin{figure*}[!t]
  \centering\small
  \begin{tabular}{ccc}
    \includegraphics[width=0.33\columnwidth]{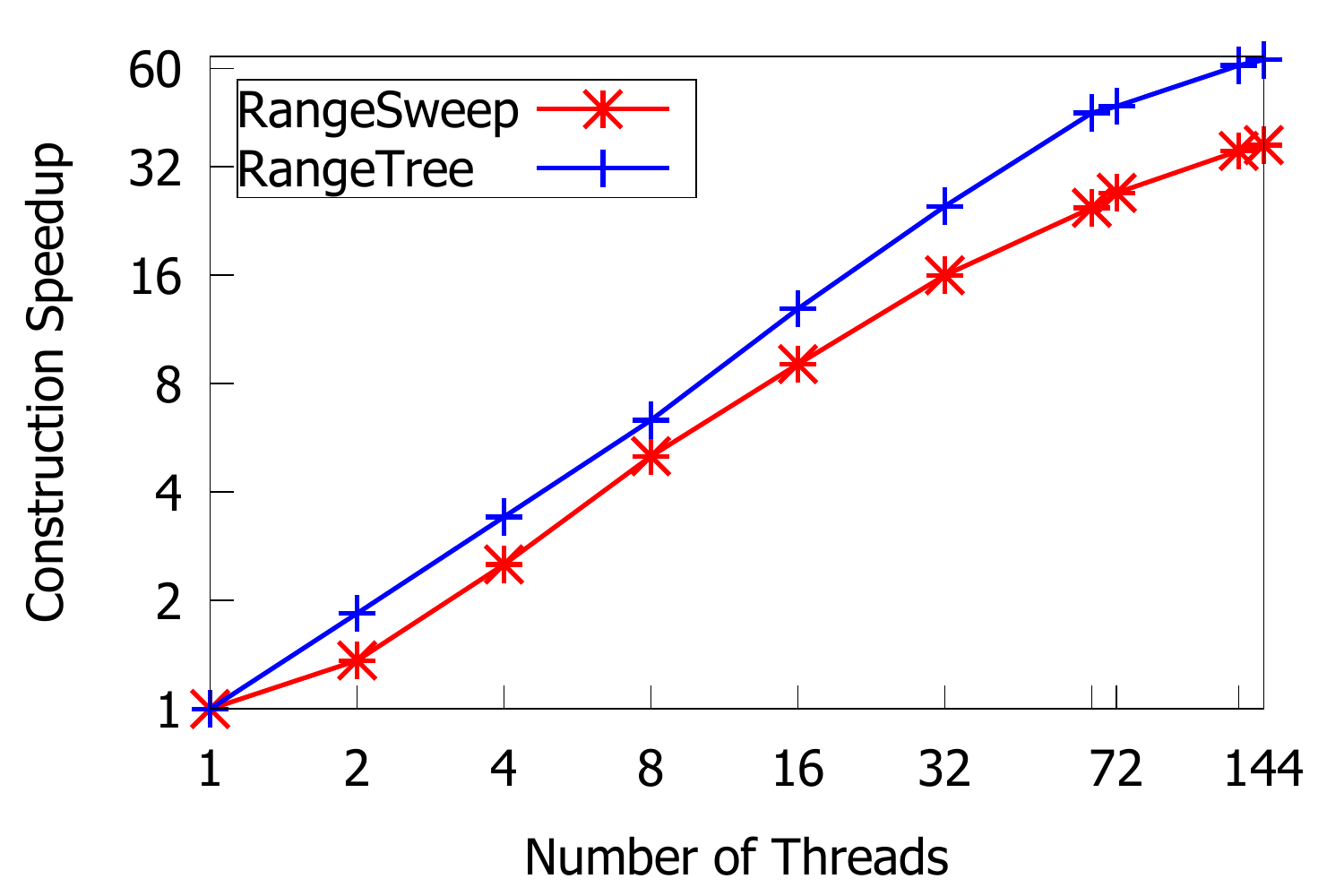} &
    \includegraphics[width=0.33\columnwidth]{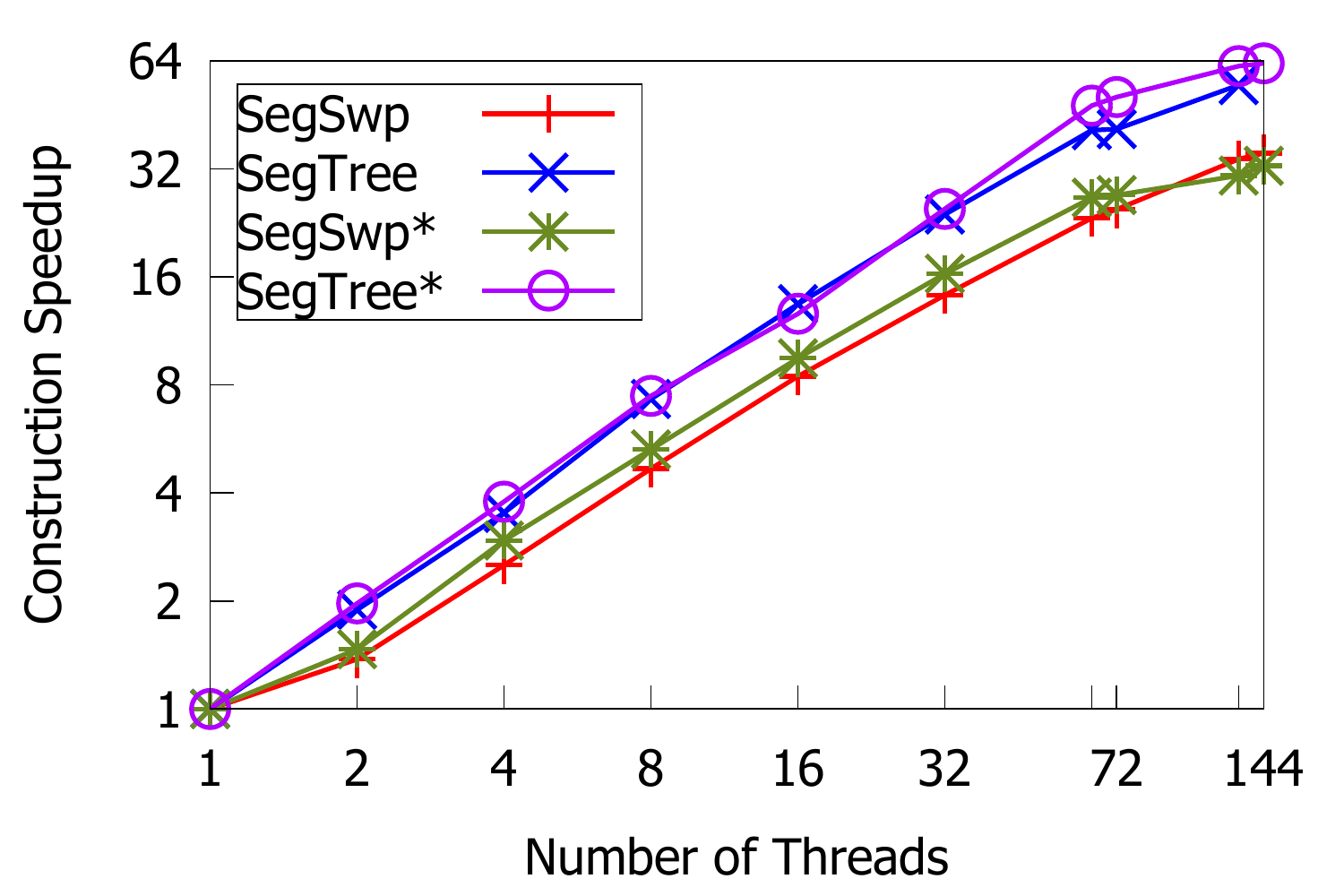}&
    \includegraphics[width=0.33\columnwidth]{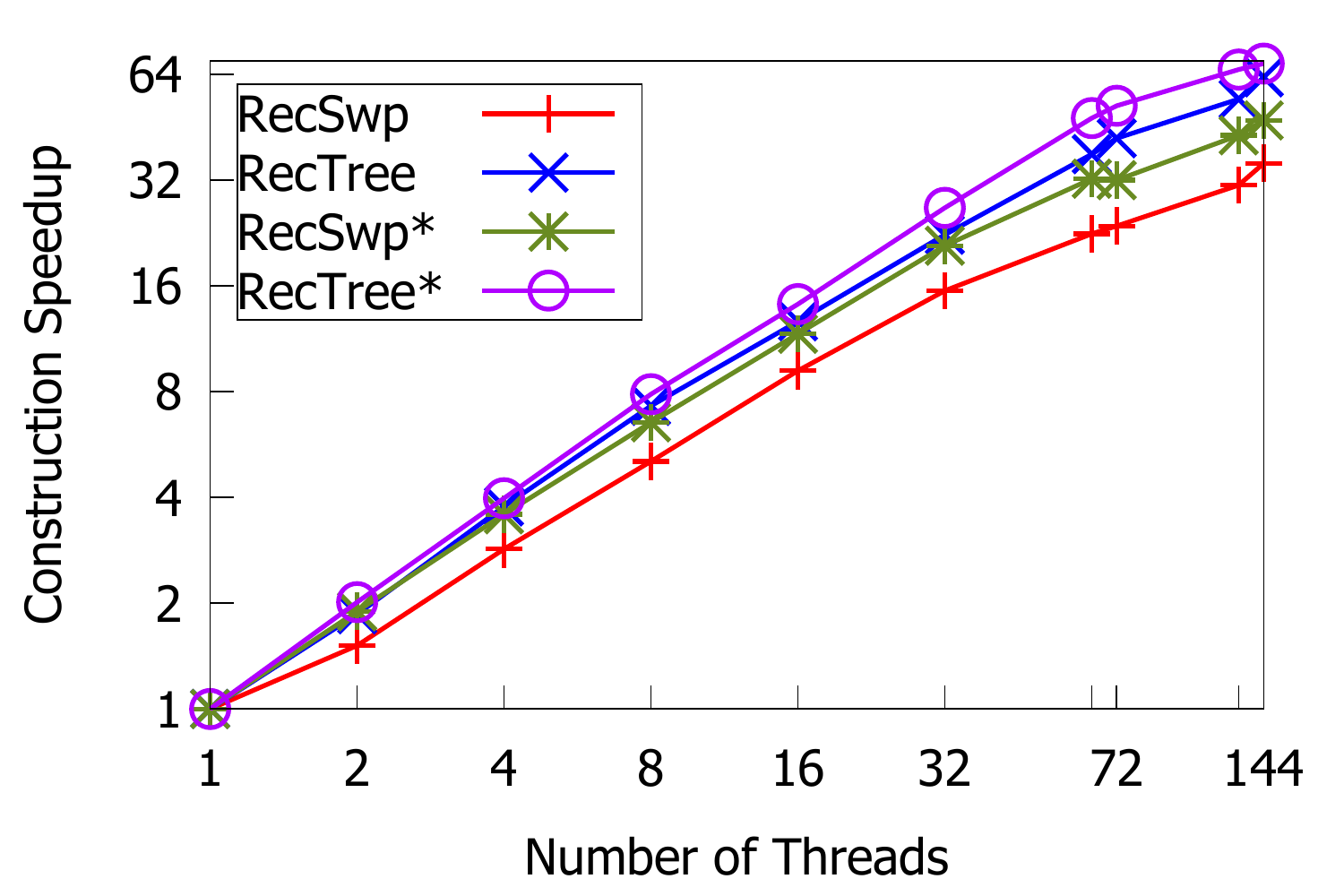}\\
    \textbf{(a).} \textbf{Range Query} &
    \textbf{(b).} \textbf{Segment Query} &
    \textbf{(c).} \textbf{Rectangle Query} \\
  \end{tabular}
  \caption{\small{The speedup of building various data structures for range and segment queries ($n=10^8$). }}\label{fig:expfigures}
\end{figure*}

\subsection{Discussion on the Data Structures for Counting Queries}
\label{app:countingexp}
We list the performance of the four data structures for fast answering counting queries in Table \ref{tab:allexpcost}. \emph{SegTree*} and \emph{SegSwp*} both have faster construction time than \emph{SegTree} and \emph{SegSwp} probably because they have simpler structure and less functionality (cannot answer list-all queries). Another reason is that \emph{SegTree*} and \emph{SegSwp*} both have smaller outer map sizes ($5\times 10^7$ vs. $10^8$), thus requiring fewer invocations of combine functions on the top level. \emph{RecSwp*} and \emph{RecTree*}, however, is about 2x slower than \emph{RecSwp} and \emph{RecTree}. This is because they have twice as large the inner tree sizes---an inner tree of a \emph{RecMap} is an interval tree storing each rectangle once as its y-interval, while an inner tree of a \emph{RecMap*} is a \countmap{} storing each rectangle twice as the two endpoints of its y-interval.

Overall, the results of these four data structures consists with the other data structures. In constructions, the sweepline algorithms have better sequential performance, but the two-level tree structures have better speedup and parallel performance. In counting queries, the sweepline algorithms are always much faster than the two-level tree structures because of their better theoretical bound.

\subsection{Scalability Discussion}
\label{app:scale}
In Figure \ref{fig:expfigures} we present the construction speedup numbers across different number of working threads for all data structures. The construction of all these data structures scale up to 144 threads. Generally, the speedup of two-level trees (about 60-70x on 144 threads) is better than the sweepline algorithms (about 30-40x on 144 threads). This is likely because of the theoretical depth ($\tilde{O}(\sqrt{n})$ vs. $O(\log^2 n)$). Another reason is that more threads means more blocks, introducing more overhead in batching and folding. As for two-level trees not only the construction is highly-parallelized, but the combine functions (\mb{Union} and \mb{Difference}) are also parallel.

\subsection{Experimental Results on Dynamic Range Trees}
\label{app:update}
We use the lazy-insert function (assuming a commutativity combine function) in PAM to support the insertion on range trees and test the performance. We first build an initial tree of size using our construction algorithm, and then conduct a sequence of insertions on this tree. We compare it with the plain insertion function (denoted as \mb{AugInsert}) in PAM which is general for the augmented tree (re-call the combine function on every node on the insertion path). We show the results in Table \ref{tab:insert}. Because the combine function takes linear time, each \mb{AugInsert} function costs about 40s, meaning that even 5 insertions may cost as expensive as re-build the whole tree. This function can be parallelized, with speedup at about 75x. The parallelism comes from the parallel combine function \mb{Union}. The \mb{LazyInsert} function is much more efficient and can finish $10^5$ insertions in 12s sequentially. Running in parallel does not really help in this case because the the combine function (\mb{Union}) is very rarely called, and even when it is called, it would be on very small tree size. When the size increases to $10^6$, the cost is also greater than rebuilding the whole tree. This means that in practice, if the insertions come in large bulks, rebuilding the tree (even sequentially) is often more efficient than inserting each point one by one into the tree. When there are only a small number of insertions coming in streams, \mb{LazyInsert} is reasonable efficient.

\begin{table}
  \centering
  \begin{tabular}{c|c|r|r|r}
    \hline
    \textbf{Algorithm} & $\boldsymbol{m}$ & \textbf{Seq.} (s) & \textbf{Par.} (s) & \textbf{Spd.} \\
    \hline
    \mb{AugInsert} & 10 & 406.079 & 5.366 & 75.670 \\
    \mb{LazyInsert} & $10^5$ & 12.506 & - & - \\
    \hline
  \end{tabular}
  \caption{\textbf{The performance of insertions on range trees using the lazy-insert function in PAM.} - ``Seq.'', ``Par.'' and ``Spd.'' mean the sequential running time, parallel running time and the speedup number respectively. }\label{tab:insert}
\end{table}

\subsection{Experimental Results on Space Consumption}
\label{app:space}
In Table \ref{tab:space}, we report the space consumption using our data structures of range and segment queries as examples to show the space-efficiency of our implementations. We note that for rectangle queries, the outer map structure is similar to segment queries, and thus can be estimated roughly using the results of segment queries. We store in each tree node a key, a value, an augmented value, the subtree size, two pointers to its left and right children and a reference counter for efficient garbage collection. Each inner tree is represented using a root pointer.

For all of them we estimate the theoretical number of nodes needed and show them in the table. The theoretical space cost is $O(n\log n)$ for all of them.
For \emph{SegTree} we use $2n\log n$ to estimate the number of inner tree nodes, and for the rest of them we simply use $n\log n$. This is because in a \emph{SegTree}, there are 4 inner trees stored in each of the outer tree node, and in the worst case, a segment can appear in at most two of them (one in the \opensets{} and the other in \symdiff{}). We compute the ratio as the actual used inner tree nodes divided by the theoretical number of inner nodes. All results in Table \ref{tab:space} are from experiments on all 144 threads.

As shown in the table, the two multi-level trees have ratio less than 100\%. This saving is mostly from the path-copying for supporting the persistence. In other words, in the process of combining the inner trees, some small pieces are preserved, and are shared among multiple inner trees. This phenomenon should be more significant when the input distribution is more skewed. In our case, because of our input is selected uniformly at random, the saving ratio is about 10\%-15\%. One special case is the \emph{SegTree}, where the ratio is only about 50\%. This is because even though theoretically in the worst case, each segment can appear in the augmented values of $O(\log n)$ outer tree nodes (one per level), in most of the cases they cancel out in the combine function. As a result, the inner tree sizes can be very small especially when the segments are short. As shown in the table, the actual number of required inner tree nodes is only about a half of the worst case, when the input endpoints are uniformly distributed.

For the sweepline algorithms, the actual used nodes are often slightly more than estimated. This is because in the parallel version of our sweepline paradigm, the trees at the beginning of each block are built separately. In the batching step, $O(n\log b)$ new tree nodes are created because of the $b$ \mb{Union} functions. In the sweeping step, $n\log n+O(n)$ new nodes are created. Because of the fold-and-sweep parallel sweepline paradigm we are using, we waste some space in the second step, when constructing the \sweepstructure{}s at the beginning of each block. As a very rough estimation, we waste about $n\log b$ (off by a small constant) nodes. In our experiment, $b=144$, $\log b\approx 7.2$, $\log n\approx 26$, which means that we may have a factor of $\log n/\log b \approx 27\%$ waste of tree nodes. This roughly matches our result of \emph{RangeSwp}. For \emph{SegSwp}, the nodes are inserted and then deleted at some point, and thus the size of the \sweepstructure{}s can be small for most of the time. In this case the wasted number of inner tree nodes is much fewer, which is only about 17\%.

In all, all of the tested data structures on range and segment tests use less than than $1.5n\log n$ tree nodes. Even the largest of them only cost 130G memory for input size $10^8$, which includes all costs of storing keys and values, as well as pointers and other information in tree nodes.

\begin{table*}
\begin{tabular}{l||r|r|r|r|r|r|r}
\hline
      & \multicolumn{1}{@{ }c@{ }|}{\textbf{\# of}} & \multicolumn{1}{@{ }c@{ }|}{\textbf{Size of an}} & \multicolumn{1}{@{ }c@{ }|}{\textbf{\# of}} & \multicolumn{1}{@{ }c@{ }|}{\textbf{Size of an}} & \multicolumn{1}{@{ }c@{ }|}{\textbf{Theoretical \#}} &       &  \\
      & \multicolumn{1}{@{ }c@{ }|}{\textbf{Outer Nodes}} & \multicolumn{1}{@{}
      c@{ }|}{\textbf{Outer Node}} & \multicolumn{1}{@{ }c@{ }|}{\textbf{Inner Nodes}} & \multicolumn{1}{@{ }c@{}
      |}{\textbf{Inner Node}} & \multicolumn{1}{@{ }c@{ }|}{\textbf{of Inner Nodes}} & \multicolumn{1}{@{ }c@{}
      |}{\textbf{Ratio}} & \multicolumn{1}{@{ }c@{}
      }{\textbf{Used Space}} \\
      & \multicolumn{1}{@{ }c@{ }|}{\textbf{($\times 10^6$)}} & \multicolumn{1}{@{ }c@{ }|}{\textbf{(Bytes)}} & \multicolumn{1}{@{ }c@{ }|}{\textbf{($\times 10^9$)}} & \multicolumn{1}{@{ }c@{ }|}{\textbf{(Bytes)}} & \multicolumn{1}{@{}
      c@{ }|}{\textbf{($\times 10^9$)}} & \multicolumn{1}{@{ }c@{ }|}{\textbf{(\%)}} & \multicolumn{1}{@{ }c@{ }}{\textbf{(G bytes)}} \\
\hline\hline
\textbf{RangeTree} & 99.97 & 48    & 2.29  & 40    & 2.56  & 89.5 & 89.75 \\
\hline
\textbf{RangeSweep} & -     & -     & 3.52  & 40    & 2.56  & 137.5 & 130.99 \\
\hline\hline
\textbf{SegTree} & 100.00 & 80    & 2.84  & 40    & 5.32  & 53.5 & 113.56 \\
\hline
\textbf{SegSweep} & -     & -     & 3.01  & 40    & 2.56  & 117.7 & 112.13 \\
\hline
\end{tabular}%
\caption{\textbf{The space consumption of our data structures on range and segment queries.} - The theoretical number of inner nodes are estimates as $2n\log n$ for \emph{SegTree}, and $n\log n$ for the rest of them. The ratio is computed as the actual used inner nodes $/$ the theoretical number of inner nodes.}\label{tab:space}
\end{table*}

\section{The PAM library}
\label{app:pam}
The PAM library~\cite{pam} uses the augmented tree structure to implement the interface of the augmented maps.
We give a list of the functions as well as their mathematical definitions in Table \ref{tab:mapfunctionsapp}. Besides the standard functions defined on maps, there are some functions designed specific for augmented maps. Some of the functions are the standard functions (e.g., \mb{Union}, \mb{Insert}) augmented with an operation $\sigma$ to combine values of entries with the same key. In the reasonable scenarios when two entries with the same key should
appear simultaneously in the map, their values would be combined by the function $\sigma$.
For example, when the same key appears in both maps involved in a \mb{Union} or \mb{Intersection},
the key will be kept as is, but their values will be combined by $\sigma$ to be the new value in the result.
Same case happens when an entry is inserted into a map that already has the same key in it, or when we build a map from a sequence that has
multiple entries with the same key.
A common scenario where this is useful for example, is to keep a count for each key, and have
\mb{Union} and \mb{Intersection} sum the counts, \mb{Insert} add the counts, and
\mb{Build} keep the total counts of the same key in the sequence. There are also some functions specific for augmented maps. For example, the \mb{ARange} function returns the augmented value of all entries within a key range,which is especially useful in answering related range queries.

The functions implemented in PAM are work-optimal and have low depth. The sequential and parallel cost of some functions are list in Table \ref{tab:mapfunctions}.

In PAM, an augmented map \texttt{aug\_map<entry>} is defined using C++ templates specifying the key type, value type and the necessary information about the augmented values. The \texttt{entry} structure should contain the follows:
\begin{itemize}
\item \textbf{typename} \texttt{K}: the key type ($K$),
\item \textbf{function} \texttt{comp}: $K\times K \mapsto$ \textbf{bool}: the comparison function on K ($<_K$)
\item \textbf{typename} \texttt{V}: the value type ($V$),
\item \textbf{typename} \texttt{A}: the augmented value type ($A$),
\item \textbf{function} \texttt{base}: $K \times V \mapsto A$: the base function ($g$)
\item \textbf{function} \texttt{combine}: $A \times A \mapsto A$: the combine function ($f$)
\item \textbf{function} \texttt{I}: $\emptyset \mapsto A$: the identity of f ($I$)
\end{itemize}

\begin{table*}
\small\centering
\begin{tabular}{>{\bf}l<{}|c|c|c}
\hline
\textit{\textbf{Function}} &\textit{\textbf{Description}} &\textit{\textbf{Work}}& \textit{\textbf{Depth}}\\
\hline
\texttt{dom}$(m)$ & $\{k(e)\,:\,e\in m\}$&$n$&$\log n$\\
\hline
\texttt{find}$(m,k)$  & \multirow{2}{*}{{$v$ \textbf{if}  $ (k,v) \in m$ \textbf{else} $\Box$}}&\multirow{2}{*}{$\log n$}&\multirow{2}{*}{$\log n$}\\
(or $m[k]$) &&&\\
\hline
\texttt{delete}$(m,k)$ & $\{ (k',v) \in m~|~k' \neq k \}$&$\log n$&$\log n$\\
\hline
\multirow{2}{*}{\texttt{insert$(m,e,\sigma)$}}&\multicolumn{1}{l|}{Argument $\sigma:V\times V\rightarrow V$}&\multirow{2}{*}{$\log n$}&\multirow{2}{*}{$\log n$}\\
\cline{2-2}
& \texttt{delete}$(m,k(e)) \cup \{ (k(e), \sigma(v(e),m[k(e)])) \}$&&\\
\hline
{\texttt{intersect}}& \multicolumn{1}{l|}{Argument $\sigma:V\times V\rightarrow V$}&\multirow{6}{*}{$n_1\log \left(\frac{n_1}{n_2}+1\right)$}&\multirow{6}{*}{$\log n_1\log n_2$}\\
\cline{2-2}
$(m_1, m_2, \sigma)$&
$\{ (k, \sigma(m_1[k],m_2[k]))|\,k \in \texttt{dom}(m_1)\cap \texttt{dom}(m_2)\}$&&\\
\cline{1-2}
\texttt{diff}$(m_1,m_2)$ & $\{e \in m_1~|~ k(e) \not\in \texttt{dom}(m_2)\}$&&\\
\cline{1-2}
\multirow{3}{*}{\texttt{union$(m_1, m_2, \sigma)$}}& \multicolumn{1}{l|}{Argument $\sigma:V\times V\rightarrow V$} &&\\
\cline{2-2}
&\texttt{diff}$(m_1,m_2)~\cup$ \texttt{diff}$(m_2,m_1)\cup$ &&\\
&\texttt{intersect}$(m_1,m_2,\sigma)$ &&\\
\hline
\multirow{3}{*}{\texttt{build}$(s, \sigma)$} & \multicolumn{1}{l|}{Arguments $\sigma:V\times V\rightarrow V, s=\langle e_1, e_2,\dots, e_n \rangle$}&\multirow{3}{*}{$n(\log n)$}&\multirow{3}{*}{$\log n$}\\
\cline{2-2}
& $\{ (k', \sigma(v'_{i_1}, v'_{i_2}, \dots)) \,|$&&\\
&$\, \exists e=(k',v'_{i_j})\in s, i_1<i_2<\dots\}$&&\\
\hline
\texttt{upTo}$(m,k)$ & $\{e \in m\,|\,k(e)\le k\}$ &\multirow{2}{*}{$\log n$}&\multirow{2}{*}{$\log n$}\\
\cline{1-2}
\texttt{range}$(m,k_1,k_2)$ & $\{e \in m\,|\,k_1\le k(e)\le k_2\}$ &&\\
\hline
\texttt{aVal}$(m)$  &$\augvalueone{m}$ &1&1\\
\hline
\texttt{aLeft}$(m,k)$  &$\augvalueone{m'} \,:\, M'=\{e'\,|\, e'\in M, k(e')<k\}$ &\multirow{2}{*}{$\log n$}&\multirow{2}{*}{$\log n$}\\
\cline{1-2}
\texttt{aRange$(m,k_1,k_2)$} &$\augvalueone{M'} \,:\, M'=\{e'\,|\,e'\in M,k_1<k(e')<k_2\}$ &&\\
\hline
\end{tabular}
\caption{\textbf{The core functions on augmented maps} - $k,k_1,k_2, k'\in K$. $v,v_1,v_2\in V$. $e\in K\times V$, $m, m_1, m_2$ are augmented maps, $n=|m|, n_i=|m_i|$, $n'$ is the output size. $s$ is a sequence.  $\Box$ represents an empty element. All bounds are in big-O notation. The bounds assume all the other functions (augmenting functions $f$, $g$ and argument $\sigma$) have constant cost.}
    \label{tab:mapfunctionsapp}
\end{table*}

\section{The Implementation of Sweepline Paradigm}
\label{app:sweepcode}
We implemented the sweepline paradigm introduced in Section \ref{sec:augsweep}. It only requires setting the list of points (in processing order) \texttt{p}, the number of points \texttt{n}, the initial \sweepstructure{} $t_0$ \texttt{Init}, the combine function ($f$) \texttt{f}, the fold function ($\phi$) \texttt{phi} and the update function ($h$) \texttt{h}.

{\ttfamily\footnotesize
\begin{lstlisting}[language=C++,frame=lines,escapechar=@]
template<class Tp, class P, class T, class F,
    class Phi, class H>
T* sweep(P* p, size_t n, T Init, Phi phi,
         F f, H h, size_t num_blocks) {
  size_t each = ((n-1)/n_blocks);
  Tp* Sums = new Tp[n_blocks];
  T* R = new T[n+1];@\vspace{.03in}@	
  // generate partial sums for each block
  parallel_for (size_t i = 0; i < n_blocks-1; ++i) {
    size_t l = i * block_size, r = l + each;
    Sums[i] = phi(p + l, p + r);}@\vspace{.03in}@	
  // Compute the prefix sums across blocks
  R[0] = Init;
  for (size_t i = 1; i < n_blocks; ++i) {
    R[i*block_size] = f(R[(i-1)*each],
        std::move(Sums[i-1])); }
  delete[] Sums;@\vspace{.03in}@	
  // Fill in final results within each block
  parallel_for (size_t i = 0; i < n_blocks; ++i) {
    size_t l = i * each;
    size_t r = (i == n_blocks - 1) ?
        (n+1) : l + each;
    for (size_t j = l+1; j < r; ++j)
      R[j] = h(R[j-1], p[j-1]);
  }
  return R;
}
\end{lstlisting}
}

\section{Using PAM for Computational Geometry Algorithms}
\label{app:codeexample}
We give some examples of our implementations using the PAM library and sweepline paradigm. We give construction code for \emph{RangeTree} and \emph{RangeSwp}, as well as the query code for \emph{RangeSwp}. They have the same inner tree structure. Note that the code shown here is almost \emph{all} code we need to implement these data structures. Comparing with all existing libraries our implementations are much simpler and as shown in the paper, are very efficient.

\para{Range Tree.}
{\ttfamily\footnotesize
\begin{lstlisting}[language=C++,frame=lines,escapechar=@]
template<typename X,typename Y>
struct RangeQuery {
  using P = pair<X,Y>;

  struct inner_map_t {
    using K = Y;
    using V = X;
    static bool comp(K a, K b) { return a < b;}
    using A = int;
    static A base(key_t k, val_t v) {return 1; }
    static A combine(A a, A b) { return a+b; }
    static A I() { return 0;} };
  using inner_map = aug_map<inner_map_t>;

  struct outer_map_t {
    using K = X;
    using V = Y;
    static bool comp(K a, K b) { return a < b;}
    using A = inner_map;
    static A base(K k, V v) {
      return A(make_pair(k.second, k.first)); }
    static A combine(A a, A b) {
      return A::union(a, b); }
    static A I() { return A();} };
  using outer_map = aug_map<outer_map_t>;
  outer_map range_tree;
	
  RangeQuery(vector<P>& p) {
    range_tree = outer_map(p); }
};

\end{lstlisting}
}

\vspace{.3in}
\para{RangeSweep.}


{\ttfamily\footnotesize
\begin{lstlisting}[language=C++,frame=lines,escapechar=@]
template<typename X,typename Y>
struct RangeQuery {
  using P = pair<X, Y>;
  using entry_t = pair<Y, X>;
  struct inner_map_t {
    using K = Y;
    using V = X;
    static bool comp(K a, K b) { return a < b;}
    using A = int;
    static A base(key_t k, val_t v) {return 1; }
    static A combine(A a, A b) { return a+b; }
    static A I() { return 0;} };
  using inner_map = aug_map<inner_map_t>;
  inner_map* ts;
  X* xs;
  size_t n;

  RangeQuery(vector<P>& p) {
    n = p.size();
    Point* A = p.data();
    auto less = [] (P a, P b)
      {return a.first < b.first;};
    parallel_sort(A, n, less);

    xs = new X[n];
    entry_t *vs = new entry_t[n];
    parallel_for (size_t i = 0; i < n; ++i) {
      xs[i] = A[i].first;
      vs[i] = entry_t(A[i].second, A[i].first); }

    auto insert = [&] (inner_map m, entry_t a) {
      return inner_map::insert(m, a);   };
    auto fold = [&] (entry_t* s, entry_t* e) {
      return inner_map(s,e);  };
    auto combine = [&] (inner_map m1, c_map m2) {
      return inner_map::union(m1, std::move(m2));};
    ts = sweep<inner_map>(vs, n, inner_map(),
                        insert, fold, combine); }

  int query(X x1, Y y1, X x2, Y y2) {
    size_t l = binary_search(xs, x1);
    size_t r = binary_search(xs, x2);
    size_t left = (l<0) ? 0 : ts[l].aug_range(y1,y2);
    size_t right = (r<0) ? 0 : ts[r].aug_range(y1,y2);
    return right-left; }
\end{lstlisting}
}